\documentclass{fundam}

\publyear{22}
\papernumber{2102}
\volume{185}
\issue{1}

\usepackage{url} 
\usepackage[ruled,lined]{algorithm2e}
\usepackage{graphicx}

\usepackage{amsmath,amssymb,amsfonts}
\usepackage{enumitem}
\usepackage{svg}
\usepackage{subfig}
\usepackage{xcolor}
\usepackage{mathtools}
\usepackage{trimclip}
\usepackage{pifont}
\usepackage{xcolor}

\newcommand{\emojiHeart}{\includegraphics[height=\fontcharht\font`A]{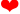} }

\newcommand{\emojiFox}{\includegraphics[height=1.3\fontcharht\font`A]{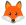} }

\newcommand{\emojiHorse}{\includegraphics[height=1.4\fontcharht\font`A]{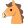} }

\begin{document}

\title{Synthesizing Petri Nets from Labelled Petri Nets\\using Token Trail Regions}

\address{Universitätsstr. 1, 58097 Hagen, Germany}

\author{Robin Bergenthum\\
Fakultät für Mathematik und Informatik \\
FernUniversität in Hagen, Germany \\
robin.bergenthum{@}fernuni-hagen.de
\and Jakub Kovář\\
Lehrgebiet Programmiersysteme \\
FernUniversität in Hagen, Germany \\
jakub.kovar@fernuni-hagen.de } 

\maketitle

\runninghead{R. Bergenthum, J. Kovář}{Synthesizing Petri Nets from Labelled Petri Nets}

\begin{abstract}
Synthesis automatically generates a process model from a behavioural specification. When the target model is a Petri net, we address synthesis through region theory. Researchers have studied region-based synthesis extensively for state-based specifications, such as transition systems and step-transition systems, as well as for language-based specifications. Accordingly, in literature, region theory is divided into two main branches: state-based regions and language-based regions.
Using state-based regions, the behavioural specification is a set of global states and related state-transitions. This representation can express conflicts and the merging of global states naturally. However, it suffers from state explosion and can not express concurrency explicitly.
Using language-based regions, the behavioural specification is a set of example runs defined by partially or totally ordered sets of events. This representation can express concurrency and branching naturally. However, it grows rapidly with the number of choices and can not express merging of conflicts.
So far, synthesis requires a trade-off between these two approaches. Both region definitions have fundamental limitations, and synthesis therefore always involves a compromise.
In this paper, we lift these limitations by introducing a new region theory that covers both state-based and language-based input. We prove that the new definition is a region meta theory that combines both concepts. It uses specifications given as a set of labelled nets, which allow us to express conflicts, concurrency and merging of local states naturally, and synthesizes a Petri net that simulates all labelled nets of the input specification.
\end{abstract}


\section{Introduction}

\newcounter{synthABabstract}
\renewcommand{\thesynthABabstract}{(\Alph{synthABabstract})}
\newcommand{\synthABabstract}{\refstepcounter{synthABabstract}\thesynthABabstract}

We model complex systems using Petri nets \cite{DBLP:journals/topnoc/AalstD13,DBLP:conf/apn/DeselJ01,DBLP:books/daglib/0001412,DBLP:books/daglib/0032298}. Petri nets have a formal semantics, an intuitive graphical representation, and the ability to express concurrency among the actions of a system. However, building a Petri net model for a real-world process remains costly and error-prone \cite{DBLP:journals/topnoc/AalstD13,DBLP:conf/icds/MayrKE07}.

Fortunately, when we model a system, we often have access to descriptions or specifications of the intended process behaviour, such as event logs, example runs, or product specifications describing use cases. We can formalize these kinds of behavioural specifications using languages, transition systems, or partially ordered languages. If a specification accurately captures the intended behaviour, we can automatically synthesize a suitable process model.

In the literature on state-based synthsis, synthesis is often concerned with deciding whether there exists a model whose behaviour exactly matches the input specification. See~\cite{DBLP:books/sp/BadouelBD15} for a comprehensive overview. In this setting, synthesis often yields a \textit{no} answer. For language-based synthesis, the synthesis problem is to compute a process model such that:
\synthABabstract\label{synthAabstract} the specification is part of the behaviour of the synthesised model, and 
\synthABabstract\label{synthBabstract} the synthesised model has minimal additional behaviour.

In this contribution, we follow the latter approach and synthesize an upper approximation of the specified behaviour. By treating the perfect-fit check as optional, synthesis always generates a model and enables many more applications in the modelling and process mining communities. 

The theory underlying Petri net synthesis is region theory \cite{DBLP:journals/acta/EhrenfeuchtR90a,DBLP:journals/acta/EhrenfeuchtR90b}. Researchers have studied region theory extensively for transition systems \cite{DBLP:books/sp/BadouelBD15}, languages \cite{DBLP:conf/concur/Darondeau98,DBLP:conf/wsc/LorenzMJ07}, and partial languages \cite{DBLP:journals/fuin/BergenthumDLM08,DBLP:conf/bpm/Bergenthum17}. The literature reports many theoretical results, concepts, case studies, and tool support, including APT \cite{DBLP:conf/ice/BestS15}, Genet \cite{DBLP:conf/acsd/CarmonaCK09}, ProM \cite{DBLP:conf/apn/vanDerWerfDHS08,DBLP:conf/icpm/Bergenthum19}, Viptool \cite{DBLP:conf/apn/BergenthumDLM08}, and I \emojiHeart Petri nets \cite{DBLP:conf/ataed/BergenthumK22}. Researchers have also extended region theory to generalised net classes, such as inhibitor nets \cite{DBLP:journals/fuin/PietkiewiczK02,DBLP:conf/pn/LorenzMB07}, to refined firing semantics, such as interval semantics \cite{DBLP:conf/ceur/PietkiewiczK23}, and recently to approaches that combine both \cite{DBLP:conf/ceur/KoutnyPK24}.

Despite the variety of existing publications, we observe across the literature that region theory follows two main branches. When the input is built from global states, for example a transition system, we apply state-based regions. In this setting, a region is a multiset of the specified states, and a finite set of places is constructed from so-called minimal regions. When the input is built from sets of ordered events, such as an event log, a language, or a partial language, we apply so-called language-based regions. Here, a region is a multiset of tokens produced by prefixes of the language, and valid places are computed by solving an associated integer linear program. To obtain a finite result, either a basis is computed or the concept of wrong continuations is used \cite{DBLP:conf/topnoc/BergenthumDM09}.

Despite their common foundations, the two branches rely on different definitions and algorithms, which forces us to choose one representation of behaviour in advance. In practical applications, this choice restricts us to one of the two techniques and often requires us to commit early to a specific view on system behaviour. As a consequence, both the modelling effort and the quality of the synthesis result strongly depend on how behaviour is captured in the first place. Important aspects such as concurrency, conflict, or state merging may be easier or harder to express depending on the chosen representation. This motivates the search for synthesis approaches that are less sensitive to a particular behavioural representation and that can flexibly integrate different perspectives on system behaviour.

Let us consider a simple workflow-like system as an example. After registration, we ask for information once, twice, or three times. In parallel, the process either saves and applies, or it stops requesting further information. After all required information has been collected and after either \textit{stop} or \textit{apply} has occurred, a final check is performed. A Petri net model of this workflow is shown in Figure~\ref{petriNet}.
We are able to construct such a model because we are experienced modellers and familiar with Petri nets. Nevertheless, the resulting model is far from simple. It contains conflicts, concurrency, a distributed initial marking, a short loop, and arc weights.

\begin{figure}
    \centering
    \includegraphics[width=0.7\linewidth]{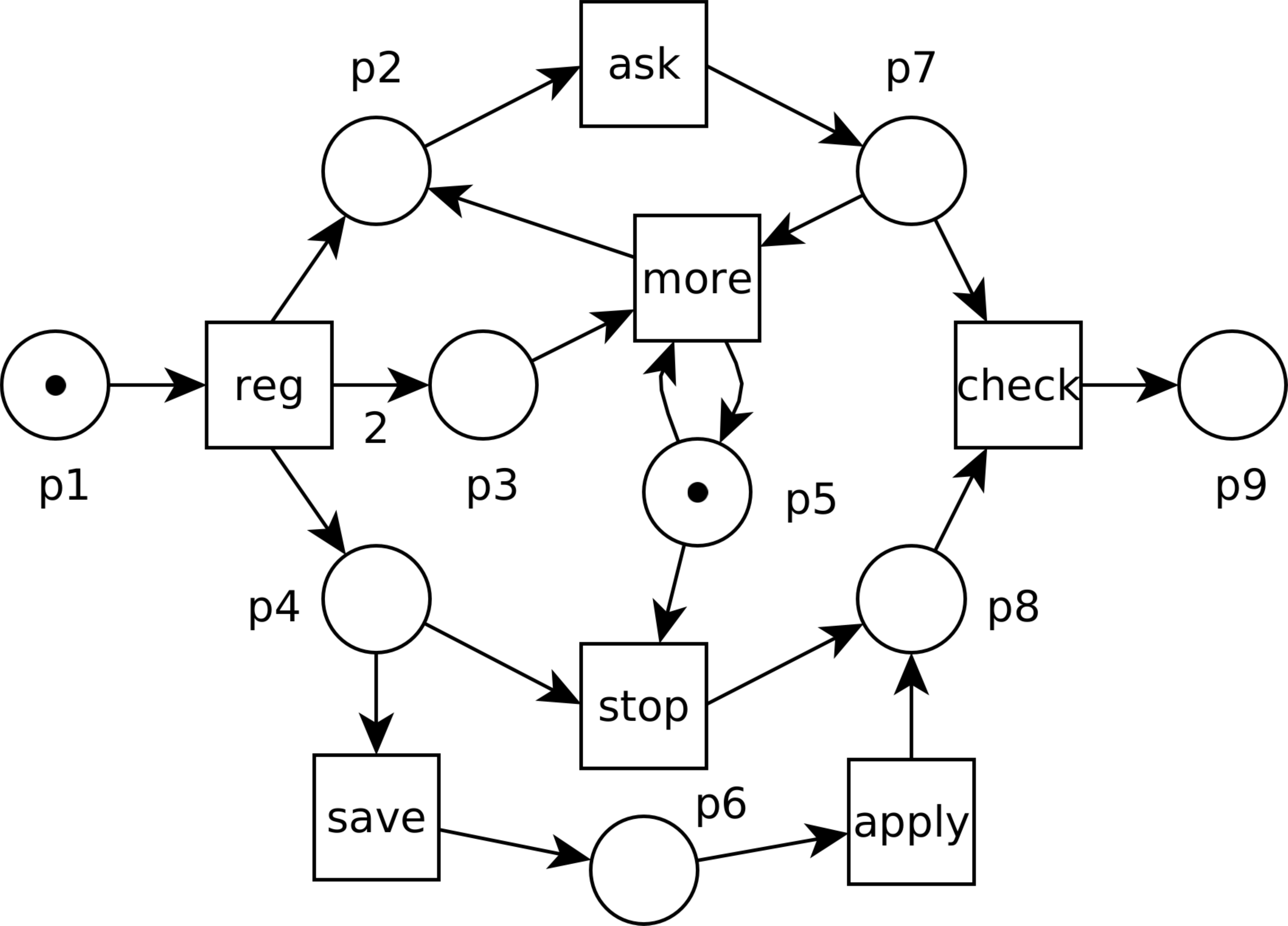}
    \caption{A Petri net model of the example workflow, used as a running example throughout the paper and exhibiting conflict, concurrency, loops, and arc weights.}
    \label{petriNet}
\end{figure}%

In many practical applications, modelling the system from scratch is simply too difficult. An alternative approach is therefore to first specify the intended behaviour of the system and then automatically generate a best-fitting process model using region theory. We assume that modelling behaviour, for example in terms of example runs or single executions, is significantly easier than constructing a complete and fully integrated system model. For the sake of our example, let us now assume that we cannot obtain Figure~\ref{petriNet} directly and instead rely on synthesis techniques to produce this result.

Using state-based region theory, we specify the intended behaviour by modelling the set of reachable states and all related state-transitions. Figure~\ref{stateGraph} shows the corresponding specification of our workflow. Although this model forms a simple state machine, it is difficult to construct because even in this small example the number of states and transitions is already large. Concurrency in the system gives rise to diamond-shaped structures in the state space, which are cumbersome to represent and reason about explicitly. In addition, we must separate states because we have to count the number of occurrences of \textit{ask}. From a modelling point of view, we would prefer to construct Figure~\ref{petriNet} directly rather than model Figure~\ref{stateGraph}.

\begin{figure}[h]
    \centering
    \includegraphics[width=0.8\linewidth]{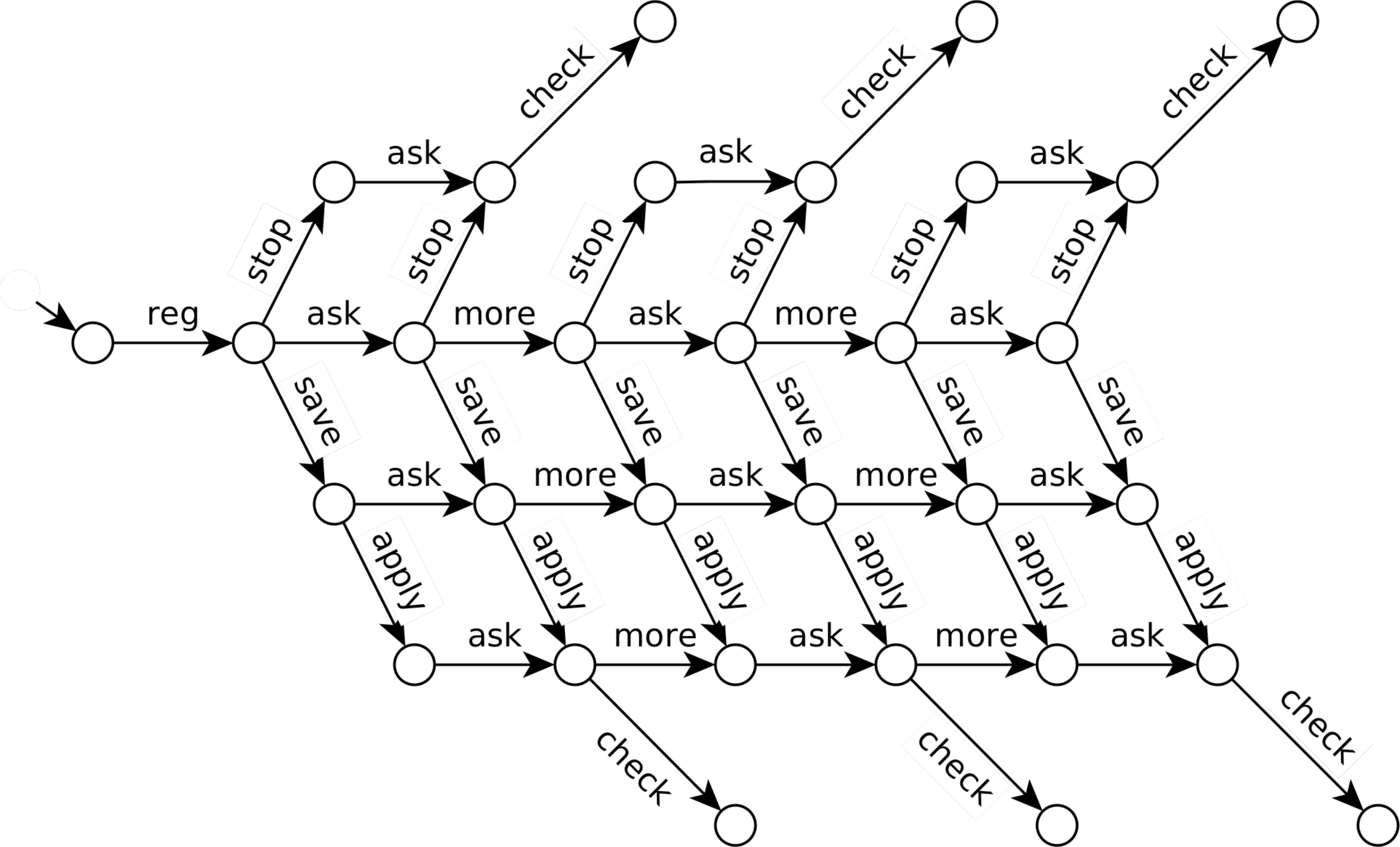}
\caption{A state-based specification of the example workflow, representing all reachable global states and possible state-transitions.}
    \label{stateGraph}
\end{figure}

Using language-based region theory, we specify the intended behaviour by modelling the set of runs of the system. A run is a conflict-free set of action occurrences, called events, together with a causal relation. Each run describes a single execution of the system and cannot express conflict. Figure \ref{pos} shows the corresponding specification of our workflow. Although each run is easy to understand, it is difficult to manage the set of all possible alternative executions. Even in this small example, we must specify six separate runs to capture the behaviour of the workflow. Adding just one more alternative at the start or end of the workflow would double the number of runs. From a modelling perspective, we maybe prefer Figure \ref{pos} over Figure \ref{stateGraph} for this workflow-like example. However, both the size of the transition system and the number of runs grow rapidly even with small increases in system model complexity.

\begin{figure}[h]
    \centering
    \includegraphics[width=\linewidth]{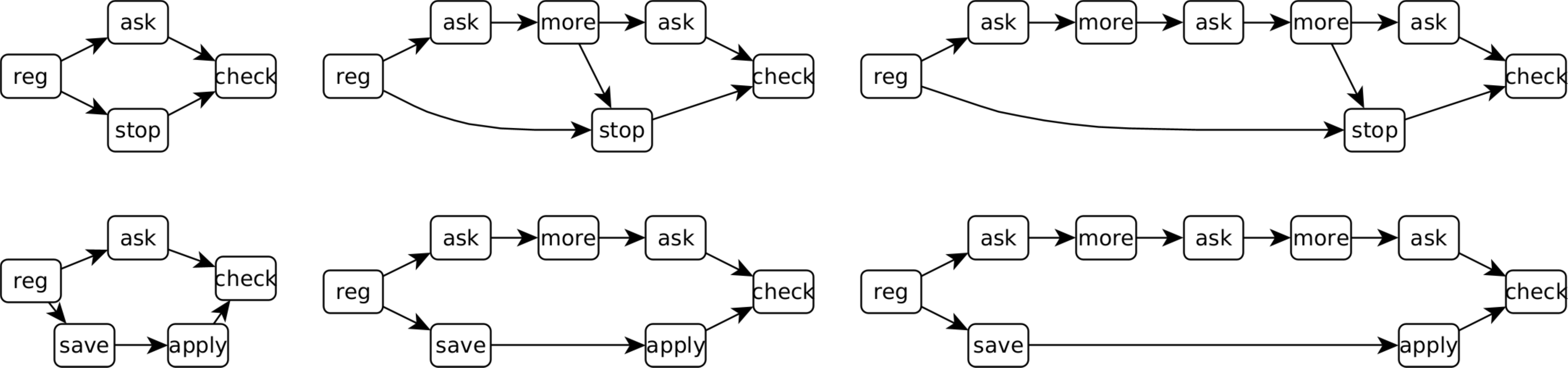}
\caption{A language-based specification of the example workflow, representing individual executions as conflict-free runs.}
    \label{pos}
\end{figure}

Up to this point, if we want to apply synthesis, we must choose either Figure \ref{stateGraph} or Figure \ref{pos} as the starting point. Figure \ref{stateGraph} grows due to concurrency, whereas Figure \ref{pos} grows because conflicts require the enumeration of alternative runs. We cannot combine these two forms of specification and cannot express concurrency and conflicts naturally at the same time.

In this paper, we address this limitation by developing a unifying meta region theory. We design this theory so that it can handle both state-based and language-based specifications. Moreover, we show that this approach can mix and combine state graphs and languages within a single specification. Furthermore, we can also consider net-based input. This is actually the key new idea of our region definition. It is based on labelled Petri nets as a specification language, allowing us to naturally express both conflict and concurrency, while still being able to split labels and build a complete specification from a set of alternative small parts of example behaviours.

Returning to our example, let us assume that we model part of the intended behaviour using the labelled net depicted in Figure~\ref{labeldnet}. The registration enables the activity \textit{ask} and introduces a choice between \textit{stop} and \textit{save}. This choice is not a free choice, since \textit{stop} can only occur after \textit{more} has been executed. As the activity \textit{ask} is performed twice in this execution, the labelled net contains two distinct events labelled \textit{ask}. Choosing between \textit{stop} and \textit{save} does not affect the subsequent behaviour, because both \textit{save} followed by \textit{apply} and \textit{stop} lead to the same local state, in which the transition \textit{check} is enabled.

Figure~\ref{labeldnet} can therefore serve as part of the specification for our workflow example. It describes a set of possible executions with branching and merging of local states, while deliberately modelling only a fragment of the overall behaviour shown in Figure~\ref{petriNet}. Label splitting is used to intuitively specify that the transition \textit{ask} occurs twice in this part of the behaviour. Note that Figure~\ref{labeldnet} is neither a state graph, nor a run, nor a process net, nor a branching process. Instead, it contains conflict, concurrency, and merging of local states. By allowing label splitting and by focusing on a partial description of the intended behaviour, Figure~\ref{labeldnet} remains considerably simpler than the complete process model in Figure~\ref{petriNet}.

\begin{figure}[h]
    \centering
    \includegraphics[
        width=0.8\linewidth,
        trim=0cm 15cm 12cm 11cm,
        clip
    ]{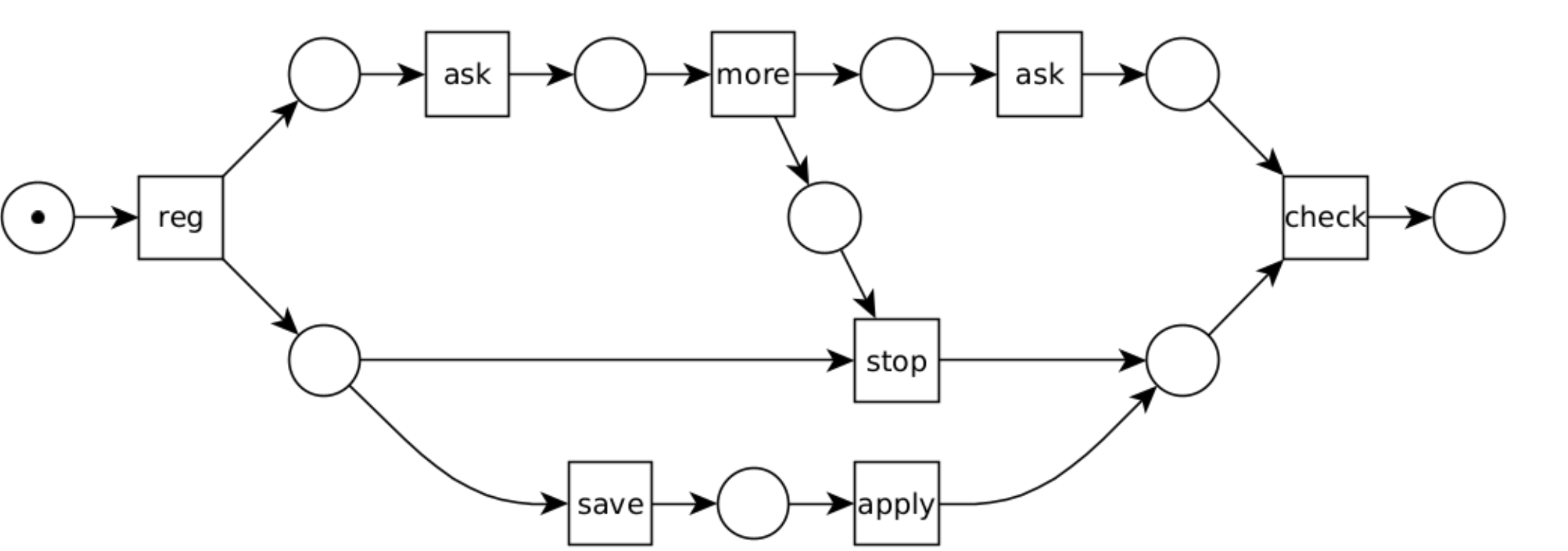}
\caption{A Petri net-based specification of the example workflow, representing part of the intended behaviour using labels and local states, and expressing conflict and concurrency.}
    \label{labeldnet}
\end{figure}

As mentioned above, Figure~\ref{labeldnet} depicts only part of the intended behaviour of our example. Similar to Figure~\ref{pos}, additional alternative runs and executions must be specified if we expect the synthesis result to include them. For this reason, we construct the combined specification shown in Figure~\ref{netSpec}. The first labelled net is a simple state graph, while the second labelled net is a simple marked graph. The third labelled net is taken directly from Figure~\ref{labeldnet}, and the fourth is a similar labelled net in which the action \textit{ask} occurs three times. We argue that this combined specification is significantly easier to devise than the complete state-based specification in Figure~\ref{stateGraph} or the full language-based specification in Figure~\ref{pos}, and it is even simpler than the brute-force approach of constructing the Petri net in Figure~\ref{petriNet} from scratch.

\begin{figure}[h]
    \centering
    \includegraphics[width=.95\linewidth]{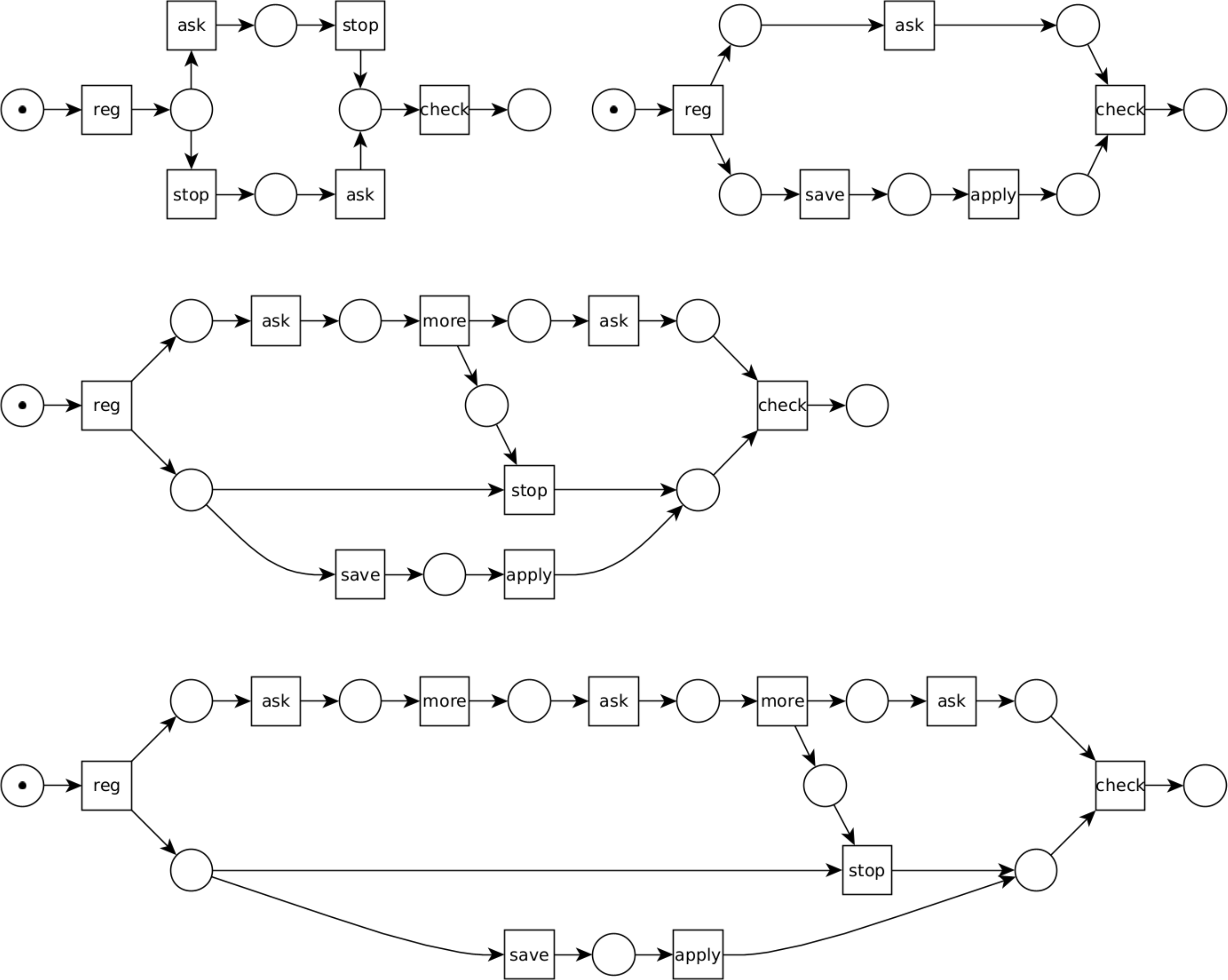}
\caption{A complete specification of the example workflow, combining state-based, language-based, and net-based models.}   \label{netSpec}
\end{figure}

We invite the reader to visit \url{www.fernuni-hagen.de/ilovepetrinets/horse/} and drag the box labelled \textit{Example~1} onto the \emojiHorse\ button. \textit{Example~1} corresponds to the specification shown in Figure~\ref{netSpec}. After a short computation, the \emojiHorse\ tool synthesises the Petri net depicted in Figure~\ref{petriNet}. The only difference is that transition \textit{reg} produces the token in the short loop place, thereby reducing the initial marking. In the \emojiHorse\ interface, the synthesis result is displayed at the top, while the labelled nets of the specification are shown below in four separate tabs.

Let us take a closer look at Figure \ref{netSpec}. The first labelled net is a simple state graph. In fact, it is a one-to-one translation of the diamond in the top left of Figure \ref{stateGraph} into a labelled net. This labelled net contains only conflict and no concurrency. Note that we could translate the entire Figure \ref{stateGraph} into a single large labelled net to obtain a complete specification of the behaviour of our example process. However, this would suffer from the same disadvantages as a purely state-based approach. 
In the \emojiHorse\ tool, this is provided as \textit{Example~2}. If we reload the tool and drag \textit{Example~2} onto the \emojiHorse\ button, we can observe both the increased size of the specification and that the synthesis result is no longer able to identify the short loop place. This is expected, as such a short loop cannot be specified using a purely global state-based specification.

The second labelled net in Figure \ref{netSpec} is a marked graph. More precisely, it is a one-to-one translation of the fourth partially order run in Figure \ref{pos} into a labelled net. This labelled net contains only concurrency and no conflict. Again, we could translate the entire Figure \ref{pos} into six separate conflict-free labelled nets to obtain a complete specification of the behaviour of our example process. However, this would again suffer from the same disadvantages as a purely language-based approach. In the \emojiHorse tool, this is provided as \textit{Example 3}. If we reload the \emojiHorse tool and drag \textit{Example 3} onto the \emojiHorse button, we see that we are now able to synthesize the short loop place again, but we observe a very notable increase in computation time. In \textit{Example 3}, which uses only partial orders, we can neither use conflicts nor merging to keep the input compact.

The third and fourth labelled net of Figure \ref{netSpec} are labelled nets with conflict, concurrency, and merging of local states. In this example, we get the third labelled net by kind of folding the partial orders two and five of Figure \ref{pos}. We get the fourth labelled net by merging the partial orders three and six of Figure \ref{pos}. However, we do not claim that labelled nets are always the better way to specify behaviour in every situation. When state-based modelling is appropriate, we use only state graphs. When language-based modelling is more suitable, we use only marked graphs. At the same time, we can mix these approaches and add labelled nets whenever appropriate. We therefore have more options. Example~4 of the \emojiHorse tool provides an additional example demonstrating that the proposed approach can even handle labelled net input containing loops and arc weights.

Altogether, we argue that a set of labelled nets provides a flexible and well-suited formalism for modelling behavioural specifications of distributed systems. Starting from such a specification, this paper addresses the synthesis problem. Note that, in our setting, we synthesise a Petri net that can simulate~\cite{DBLP:conf/ijcai/Milner71,DBLP:conf/tcs/Park81,DBLP:journals/iandc/JoyalNW96} all labelled nets in the specification, which is crucial to our contribution. In state-based synthesis, the input is a subgraph of the reachability graph of the synthesised Petri net. In language-based synthesis, the input is a subset of the language of the synthesised Petri net. Thus, the underlying semantics differ. Our region definition unifies both views by requiring that the synthesised net can simulate all inputs. That is, there exists a linear mapping from all reachable states of any input labelled net to the reachable states of the synthesis result. This mapping respects the initial states and all step transitions in the state graph of the input. If we restrict the specification to state graphs or to marked graphs, we directly obtain the classical state-based or language-based synthesis formulations. If the input is purely state-based, it is a subgraph of the reachability graph of the result, the simulation mapping is the identity. For language-based synthesis, there is the non-trivial result that step semantics and partial order semantics match \cite{DBLP:books/sp/Vogler92,DBLP:conf/apn/JuhasLD05}. If the model can step as the specification, the input is part of the language of the synthesis result.

Observe that, we require simulation, not bisimulation, because each input net models only a part of the intended behaviour. Thus, requiring each input net to be able to simulate the synthesis result would be a mistake. 

The new region definition builds on the token trail semantics introduced in~\cite{DBLP:conf/ataed/BergenthumK22,DBLP:conf/apn/BergenthumFK23,DBLP:conf/apn/KovarB24}. This paper is an extended version of the conference contribution~\cite{DBLP:conf/apn/BergenthumK25}. Using token trails, we define a corresponding region theory in Chapter~3 and demonstrate that labelled nets provide a natural and compact way to model behavioural specifications. In Chapter~4, we show that this theory is consistent with both state-based and language-based regions and therefore constitutes a new meta region theory. In Chapter~5, we present how to compute a synthesis result using an integer linear program and introduce the corresponding implementation in the \emojiHorse\ tool. In Chapter~6, we discuss applications of the approach, focusing in particular on process discovery.

\section{Preliminaries}

Let $\mathbb{N}$ be the non-negative integers. Let $f$ be a function and $B$ be a subset of the domain of $f$. We write $f|_B$ to denote the restriction of $f$ to $B$. As usual, we call a function $m: A \to \mathbb{N}$ a multiset and write $m = \sum_{a\in A}m(a) \cdot a$ to denote multiplicities of elements in $m$. Let $m': A \to \mathbb{N}$ be another multiset. We write $m \leq m'$ iff $\forall a \in A : m(a) \leq m'(a)$ holds. We model distributed systems by Petri nets \cite{DBLP:journals/topnoc/AalstD13,DBLP:conf/apn/DeselJ01,DBLP:books/daglib/0001412,DBLP:books/daglib/0032298}.

\begin{definition}
    A Petri net is a tuple $(P, T, W)$ where $P$ is a finite set of places, $T$ is a finite set of transitions so that $P \cap T = \emptyset$ holds, and $W: (P \times T) \cup (T \times P) \to \mathbb{N}$ is a multiset of arcs. A marking of $(P, T, W)$ is a multiset $m:P \to \mathbb{N}$. Let $m_0$ be a marking, we call $N = (P, T, W, m_0)$ a marked Petri net and $m_0$ the initial marking of $N$.
\end{definition}

Figure \ref{petriNet} depicts a marked Petri net. Transitions are shown as rectangles, places as circles, the multiset of arcs as a set of weighted arcs, and the initial marking as a set of black dots, referred to as tokens.
For Petri nets, there is a firing rule. Let $t$ be a transition of a marked Petri net $(P, T, W, m_0)$. We denote $\circ t = \sum_{p \in P}W(p,t) \cdot p$ the weighted preset of $t$. We denote $t \circ = \sum_{p \in P}W(t,p) \cdot p$ the weighted postset of $t$. A transition $t$ can fire in marking $m$ iff $m \geq \circ t$ holds. Once transition $t$ fires, the marking of the Petri net changes from $m$ to $m' = m - \circ t + t \circ$.
In our example marked Petri net, transition \textit{reg} can fire in the initial marking. If \textit{reg} fires, it removes one token from $p_1$ and produces a new token in $p_2$, two new tokens in $p_3$, and a new token in $p_4$. In this new marking transitions \textit{ask}, \textit{stop}, and \textit{save} can fire. \textit{reg} is not enabled any more, because there are no more tokens in $p_1$. 

Repeatedly applying the firing rule produces so-called firing sequences. These firing sequences are the most basic behavioural model of Petri nets. For example, the sequence \textit{reg ask stop check} is enabled in the marked Petri net of Figure \ref{petriNet}. The sequence \textit{reg save apply ask check} is another example. Let $N$ be a marked Petri net, the set of all enabled firing sequences of $N$ is the sequential language of $N$.

\newpage
Another formalism used to model the behaviour of a Petri net is the reachability graph. A marking is reachable if there is a firing sequence that produces this marking.

\begin{definition}
    Let $N = (P, T, W, m_0)$ be a marked Petri net. The reachability graph of $N$ is a tuple $(R,T,X,m_0)$, where $R$ is the set of reachable markings of $N$, $T$ is the set of transitions of $N$, and $X \subseteq R \times T \times R$ is a set of labelled state-transitions such that $(m,t,m')$ is in $X$ if and only if $t$ is enabled in $(P,T,W,m)$, and firing $t$ in $m$ leads to the marking $m'$. We call a tuple $(R',T',X',i)$ a state graph enabled in $N$ if there is an injective function $g: R' \to R$, $g(i) = m_0$, $T' \subseteq T$, $\forall (m,t,m') \in X' : (g(m),t,g(m')) \in X$, and for every $m' \in R'$ there is a directed path from $i$ to $m'$ using the elements of $X'$ as arcs. We call the set of enabled state graphs the state language of $N$.
\end{definition}

Roughly speaking, every node of a state graph relates to a reachable state, we don’t have to include all state-transitions and states if every node can be reached from the initial node. Thus, a state graph is kind of a prefix of a reachability graph. 
Figure \ref{stateGraph} depicts a state graph modelling the behaviour of the marked Petri net depicted in Figure 1. The state graph has $31$ states and $43$ state-transitions labelled with transitions of the Petri net. The state graph describes the Petri nets behaviour as follows. At first, we must fire transition \textit{reg}. Then, we have some choices. We can execute \textit{stop}, \textit{ask}, or \textit{save}. We have to count occurrences of the action \textit{ask} and after \textit{stop} there is no \textit{more}. Finally, firing \textit{check} always leads to a deadlock. After firing \textit{check}, there are either two, one, or no tokens in $p_3$, and one or no tokens in $p_5$.
The state language includes firing sequences as the set of all paths through the graph. In this sense, state graphs can merge firing sequences at shared states and may contain loops. However, these graphs are not able to directly express concurrency.
Firing \textit{reg} in the initial marking depicted in Figure \ref{petriNet} leads to the marking $p_2 + 2 \cdot p_3 + p_4 + p_5$. In this marking, transitions \textit{ask} and \textit{stop}, or \textit{ask} and \textit{save} can fire concurrently because they don’t share tokens. Neither firing sequences nor state graphs can express this concurrency. Therefore, there are additional semantics for Petri nets in the literature that can explicitly express concurrency. These include step semantics of Petri nets \cite{DBLP:journals/fuin/Grabowski81}, process net semantics of Petri nets \cite{DBLP:journals/iandc/GoltzR83}, token flow semantics of Petri nets \cite{DBLP:conf/apn/JuhasLD05}, and compact token flow semantics of Petri nets \cite{DBLP:journals/fuin/BergenthumL15}. Fortunately, these semantics are equivalent \cite{DBLP:journals/jip/Kiehn88,DBLP:books/sp/Vogler92,DBLP:conf/apn/JuhasLD05,DBLP:journals/fuin/BergenthumL15} and all define the same partial language. In a partial language, every so-called run is a partially ordered set of events. Obviously, runs can express concurrency and provide an intuitive approach to modelling the behaviour of a distributed system.

Using compact token flow semantics, we can decide if a run is in the partial language of a Petri net in polynomial time \cite{DBLP:conf/apn/Bergenthum21}. Roughly speaking, a compact token flow is a distribution of tokens on the arcs of a run so that every event receives enough tokens, no event must pass on too many tokens, and all events share tokens from the initial marking.

\begin{definition}
    \label{defTokenFlows}
    Let $T$ be a set of labels. A labelled partial order is a triple $(V,\ll,l)$ where $V$ is a finite set of events, ${\ll} \subseteq V \times V$ is a transitive and irreflexive relation, and the labelling function $l: V \to T$ assigns a label to every event. A run is a triple $(V,<,l)$, with a relation ${<} \subseteq V \times V$, iff its irreflexive transitive closure $(V,<^*,l)$ is a labelled partial order. Let $N = (P, T, W, m_0)$ be a marked Petri net and $R = (V,<,l)$ be a run so that $l(V) \subseteq T$ holds. Let $\blacktriangleright,{\scriptstyle\blacksquare} \not \in V$ be two symbols. A compact token flow is a function $x: ((\{\blacktriangleright\}\times V)\cup{<}\cup(V \times\{{\scriptstyle\blacksquare}\})) \to \mathbb{N}$. Let $v \in V$ be an event. We denote $in_x(v) := x(\blacktriangleright,v) + \sum_{v'<v}x(v',v)$ the inflow of $v$, and $out_x(v) := \sum_{v<v'}x(v,v') + x(v,{\scriptstyle\blacksquare})$ the outflow of $v$. We define, $x$ is valid for $p \in P$ iff the following conditions hold:

    \begin{enumerate}[label=(\roman*)]
        \item \label{ctfi} $\forall v \in V: in_x(v) \geq W(p,l(v))$,
        \item \label{ctfii} $\forall v \in V: out_x(v) = in_x(v) + W(l(v),p) - W(p,l(v))$, and
        \item \label{ctfiii} $\sum_{v \in V}x(\blacktriangleright,v) = m_0(p)$.
    \end{enumerate}

    $R$ is enabled in $N$ iff there is a valid compact token flow for every $p \in P$. The set of all enabled runs of $N$ is the partial language of $N$.   
\end{definition}

Figure \ref{pos} depicts six different runs modelling the behaviour of the marked Petri net depicted in Figure \ref{petriNet}. Every run starts with executing transition \textit{reg}. The first run models the concurrent execution of transitions \textit{ask} and \textit{stop} before firing \textit{check}. The second run models the execution of the loop \textit{ask more} before concurrently executing \textit{ask} and \textit{stop}. Transition \textit{stop} can only occur after \textit{more} because of place $p_5$. The third run models two times the loop. Runs four, five, and six model the execution of the loop of \textit{ask} and \textit{more} in parallel to the execution of the sequence \textit{save apply}. Figure \ref{compactFlows} depicts four valid compact token flows for the places $p_1$, $p_2$, $p_3$, and $p_5$ of Figure \ref{petriNet} in the last run of Figure \ref{pos}. For every token flow, the related place defines the number of tokens each event must receive and the number of tokens an event can pass on to later events. A compact token flow is a distribution of tokens over the arcs of the run, respecting the demand and capability of each event. In addition, a compact token flow distributes the initial marking and collects any superfluous tokens to define the final marking.

\begin{figure}[b]
    \centering
    \includegraphics[width=\linewidth]{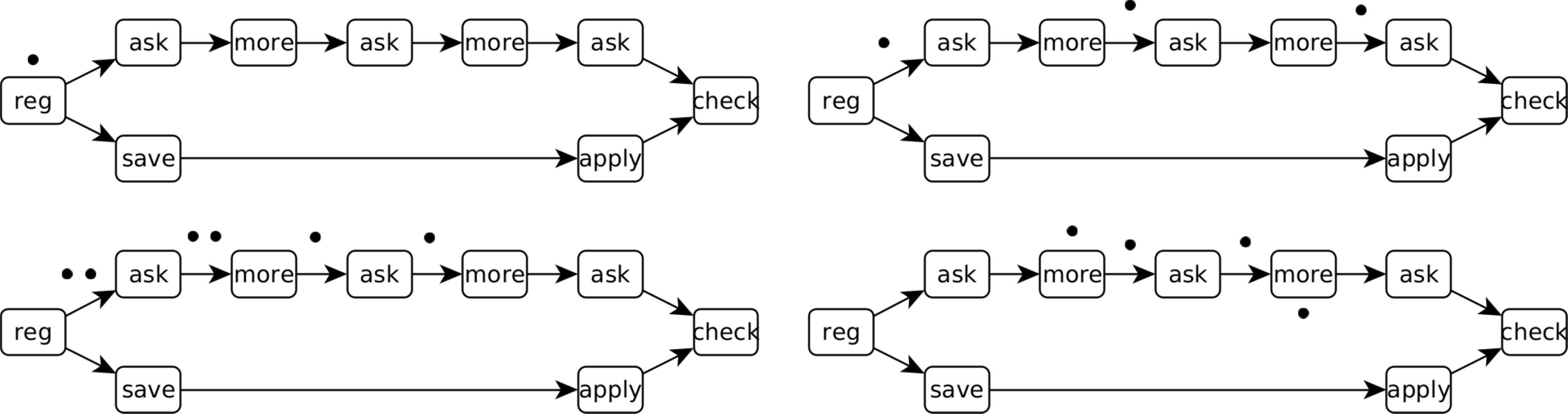}
    \caption{Four valid compact token flows for places $p_1$, $p_2$, $p_3$, and $p_5$ in the last run of Figure \ref{pos}.}
    \label{compactFlows}
\end{figure}

In Figure \ref{compactFlows}, in the first token flow, related to place $p_1$, \textit{reg} consumes from the initial marking. We show this token above the event. In the second token flow, related to place $p_2$, \textit{reg} produces one token for the first \textit{ask}. Then every occurrence of \textit{more} produces a token for the following \textit{ask}. In the third token flow, related to place $p_3$, \textit{reg} produces two tokens, one for every occurrence of \textit{more}. Here we clearly see how the token flow is passed to later events. In the fourth token flow, related to place $p_5$, \textit{more} consumes from the initial marking and produces a token for the second occurrence of \textit{more}. This \textit{more} produces one token for the final marking. We show this token below the event. Figure \ref{pos} has six runs, and Figure \ref{petriNet} has nine places. To show that Figure \ref{pos} is in the partial language of Figure \ref{petriNet}, we need $54$ valid compact token flows.

Altogether, Figure \ref{stateGraph} and Figure \ref{pos} model the behaviour of the Petri net in Figure \ref{petriNet}. Figure \ref{stateGraph} is unable to express concurrency of transitions. Figure \ref{pos} requires six separate runs, because runs cannot contain conflicts. Thus, there is always some trade-off when choosing one semantics over the other.

In the literature, we find other, possibly more advanced techniques for modelling states or executions of Petri nets. Potential candidates include process nets \cite{DBLP:journals/iandc/GoltzR83}, branching processes \cite{DBLP:journals/iandc/GoltzR83,DBLP:books/daglib/0032298}, prime event structures \cite{DBLP:conf/apn/Winskel86}, oclets \cite{DBLP:conf/apn/Fahland09}, and spread nets \cite{DBLP:journals/jlamp/PinnaF20}. These formalisms build upon state graphs and runs but add steps and/or conditions and cuts to enhance the expressiveness of the modelling techniques. However, there is a drawback: as the technique becomes more advanced, modelling behaviour becomes trickier. Specifying all possible steps, combinations of different preconditions, and other factors must be taken into account. Furthermore, by expanding but complicating the semantics, these techniques become just another high-level modelling language to be learned. Although the semantics can often explicitly express conflict and concurrency, they are still not able to intuitively merge local states from different executions. A branching process, for example, is called a branching process because it can only branch, not merge. If we specify behaviour in terms of a branching process, the model will fan out. For this reason, the main application of such techniques is to calculate the complete behaviour of a given model. It is difficult to specify behaviour in terms of these types of structures. To merely illustrate the idea of this problem, Figure \ref{branching} depicts the branching process of Figure \ref{petriNet}. Obviously, we do not want to specify a process in terms of Figure \ref{branching} in order to synthesise Figure \ref{petriNet}.

\begin{figure}[h]
    \centering
    \includegraphics[width=0.85\linewidth]{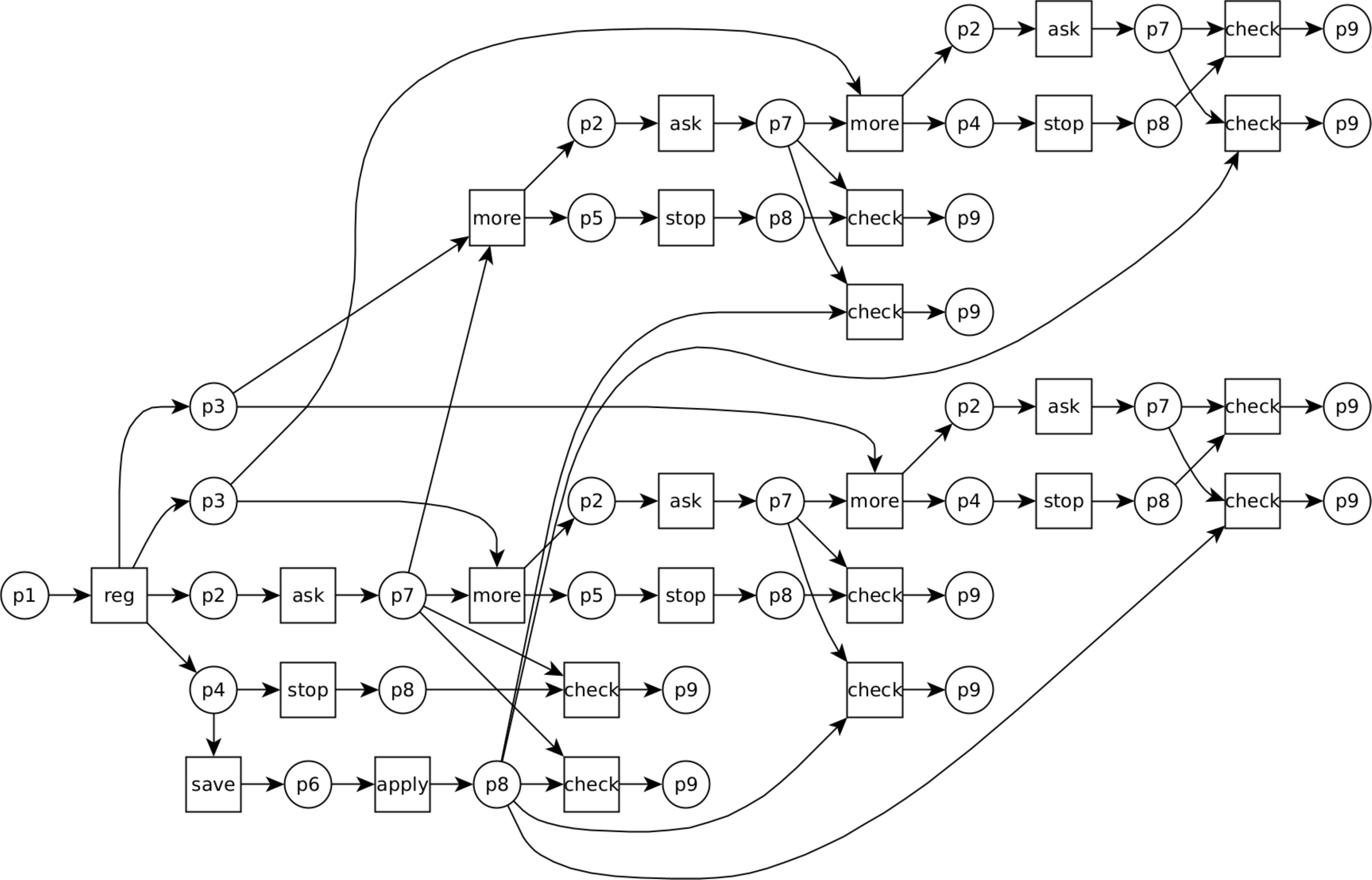}
    \caption{The branching process of Figure \ref{petriNet}.}
    \label{branching}
\end{figure}

To address the problem of intuitively modelling behaviour of a system and to be capable of expressing concurrency, conflict, and merging of local states, we introduced the so-called net language of a Petri net \cite{DBLP:conf/apn/BergenthumFK23,DBLP:conf/apn/KovarB24}. The net language of a Petri net model is the set of all labelled nets such that the model can simulate every labelled net. In our context, simulate means that there is a mapping from every state of the language to every state of the model, respecting the initial marking and preserving the state-transition behaviour of enabled multisets of transitions. Roughly speaking, the model can behave just like the specification. Naturally, the net language also includes all enabled firing sequences, state graphs, runs, and branching processes. For these formalisms there is a simple one-to-one translation into a labelled net. It is not surprising that it is easy to prove that the resulting labelled nets are in the net language as well. On the other hand, the net language is not imprecise. The set of all step sequentializations of all the nets in the net language is the step language of the model. In this sense, there is no additional behaviour in the language that is not justified by the model. Furthermore, we already know how to model in terms of Petri nets, so there is no need to learn yet another modelling language. For details, we refer the reader to \cite{DBLP:conf/apn/BergenthumFK23,DBLP:conf/apn/KovarB24}, and simply state the following formal definition of the net language.
The main idea of this definition is to introduce the linear mapping required to prove simulation and to express this mapping as a marking of the labeled net. Like compact token flows, the marking can be understood as a distribution of tokens over the local places of the specification.

\begin{definition}
    \label{defTokenTrails}
    Let $T$ be a set of labels. A labelled net is a tuple $(C,E,F,i_0,l)$ where $(C,E,F,i_0)$ is a marked Petri net and $l: E \to T$ a labelling function assigning a label to every transition of the labelled net.
    Let $N = (P, T, W, m_0)$ be a marked Petri net and $L = (C,E,F,i_0,l)$ be a marked labelled net. A token trail $x: C \to \mathbb{N}$ is a marking of $L$. Let $e \in E$ be a labelled transition. We denote $in_x(e) := \sum_{(c,e)\in F}F(c,e) \cdot x(c)$ the inflow of tokens to $e$ in $x$, and the weighted sum $out_x(e) := \sum_{(e,c)\in F}F(e,c) \cdot x(c)$ the outflow of tokens from $e$ in $x$. We define, $x$ is valid for $p \in P$ iff the following conditions hold:

    \begin{enumerate}[label=(\Roman*)]
        \item \label{ttI} $\forall e \in E: in_x(e) \geq W(p,l(e))$,
        \item \label{ttII} $\forall e \in E: out_x(e) = in_x(e) + W(l(e),p) - W(p,l(e))$, and
        \item \label{ttIII} $\sum_{c \in C}i_0(c) \cdot x(c) = m_0(p)$.
    \end{enumerate}

    $L$ is enabled in $N$ iff there is a valid token trail for every $p \in P$. The set of all enabled labelled nets of $N$ is the net language $\mathcal{L}(N)$ of $N$.
\end{definition}

\begin{figure}[h]
    \centering
    \includegraphics[width=0.90\linewidth]{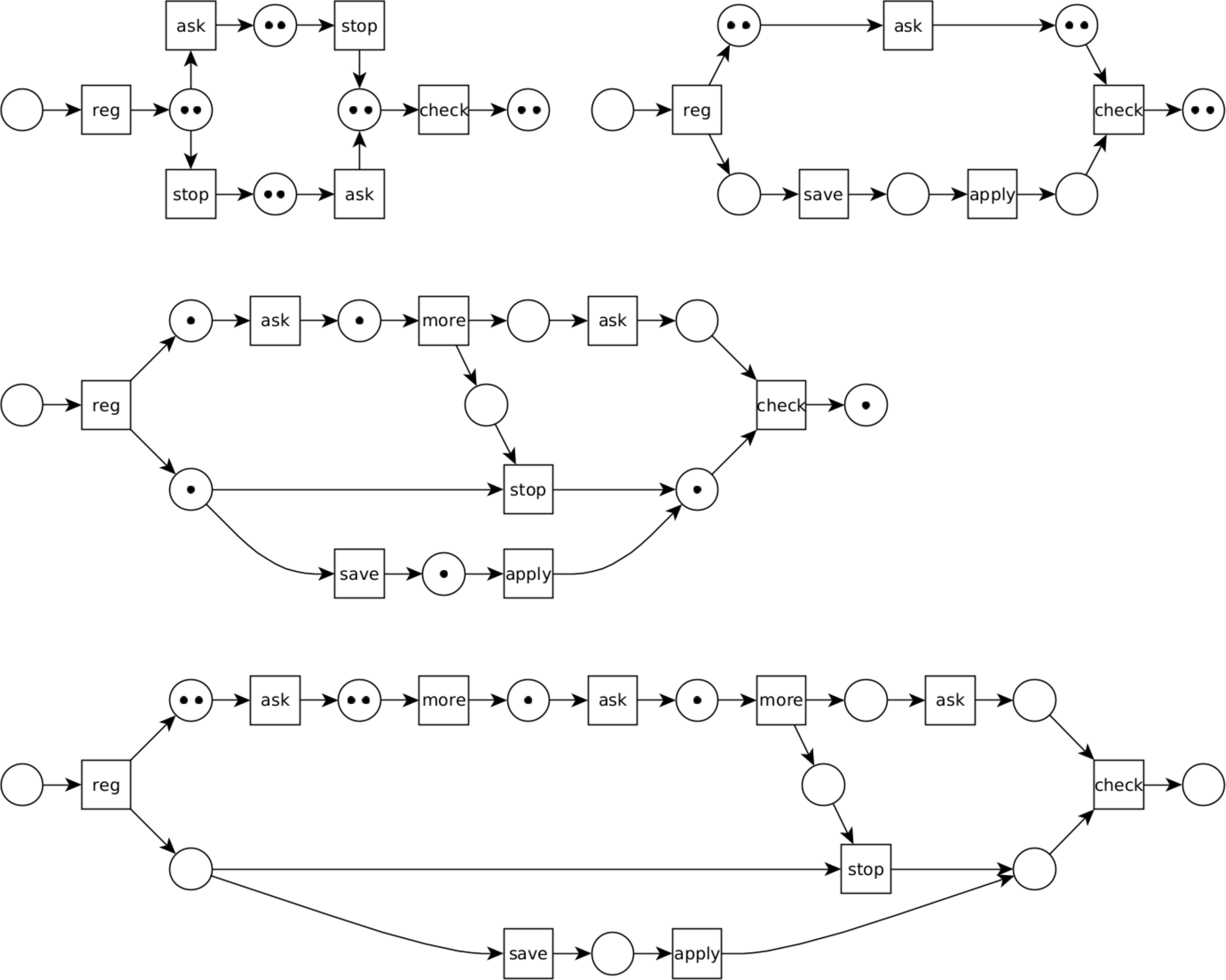}
    \caption{Four valid token trails for $p_3$ of Figure \ref{petriNet} in the four nets of Figure \ref{netSpec}.}
    \label{tokenTrails}
\end{figure}

Figure \ref{tokenTrails} depicts four valid token trails, one in every labelled net of the specification of Figure \ref{netSpec} for the place $p_3$ of Figure \ref{petriNet}. A token trail is a marking and, just like a token flow, describes a valid distribution of tokens respecting the demand and capability of each event. A valid token trail proves that, the respective place of the Petri net is able to simulate the labelled net. In this example, regarding place $p_3$, only transitions labelled \textit{reg} can produce tokens. Thus, in every token trail, the sum of tokens in the postset of transitions labelled \textit{reg}, minus the sum of tokens in the preset of the same transition, is exactly $2$. For every transition labelled \textit{more} there is one less token in the postset than in the preset, because \textit{more} is consuming one token from $p_3$ in the Petri net model. All the other labels do not interact with $p_3$, the sum of ingoing tokens is the same as the sum of outgoing tokens.

If we think of a token trail as a distribution of tokens, in the first net of Figure \ref{tokenTrails}, we see that if there is a conflict, events kind of share tokens. If there is a merge, events agree on the number of tokens passed. In this example, the path \textit{reg ask stop} and the path \textit{reg stop ask} will both provide two tokens. In the second net of Figure \ref{tokenTrails}, we see that if there is a split, events distribute tokens, and if there is a join, events collect tokens. But then again, tokens in a token trail model local states, they don't have to relate to individual tokens of the Petri net. This is why it is much easier to specify behaviour in terms of labelled nets than in terms of process nets or branching processes. Here, we also see that there might be multiple valid token trails for one place because the tokens can move via \textit{ask} or via \textit{save apply}. In nets three and four, there is a mix of splits, joins, conflicts, and merges of tokens. In the third net, one of the tokens is passed to and consumed by \textit{more}. In the fourth net, every \textit{more} receives and consumes one token produced by \textit{reg}.

To prove that the four nets of Figure \ref{netSpec} are in the net language of Figure~\ref{petriNet}, we also have to find valid token trails for the other eight places of Figure~\ref{petriNet}. Here, we refer the reader to our web tool at \url{www.fernuni-hagen.de/ilovepetrinets/fox/}, where we have prepared every labbeled net of this example to be dragged to the \emojiFox icon. Clicking on places in the Petri net will display the related valid token trails in the labelled nets.
For more examples and proofs that token trails indeed cover the state language and the partial language, and that every net of the net language can be simulated by the Petri net, we refer the reader to \cite{DBLP:conf/apn/BergenthumFK23,DBLP:conf/apn/KovarB24}. Note that the labelled nets in a net language may also contain loops, arc weights, and distributed initial markings.

\section{Token Trail Regions}

\newcounter{synthAB}
\renewcommand{\thesynthAB}{(\Alph{synthAB})}
\newcommand{\synthAB}{\refstepcounter{synthAB}\thesynthAB}

In this section, we address the synthesis problem for a set of labelled nets. The synthesis problem is computing a Petri net from a set of labelled nets such that: 
\synthAB\label{synthA} every labelled Petri net is in the net language of the synthesised Petri net, and \synthAB\label{synthB} the generated model has minimal additional behaviour.
Like mentioned above, to solve this problem. we strictly follow the general ideas of region theory. 

\begin{definition}\label{defRegion}
    Let $S = \{(C_1,E_1,F_1,i_1,l_1),\ldots,(C_n,E_n,F_n,i_n,l_n)\}$ be a set of marked labelled nets. A token trail region is a marking $r:\bigcup_i C_i \to \mathbb{N}$ iff the following two conditions hold:

    \begin{enumerate}[label=(\Roman*)]
        \setcounter{enumi}{3}
        \item \label{regionIV} $\forall \nu, \mu \in \mathbb{N}, e \in E_\nu, e' \in E_\mu, l_\nu (e) = l_\mu(e'): out_r(e) - in_r(e) = out_r(e') - in_r(e')$ and
        \item \label{regionV} $\forall \nu, \mu \in \mathbb{N}: \sum_{c\in C_\nu}i_\nu(c)\cdot r(c) = \sum_{c'\in C_\mu}i_\mu(c')\cdot r(c')$.
    \end{enumerate}

    For every labelled transition $e$, we call $out_r(e) - in_r(e)$ the rise of $e$. Using this notion, property \ref{regionIV} ensures: same label, same rise. For every labelled net, we call $\sum_{c\in C_\nu}i_\nu(c)\cdot r(c)$ the initial sum of tokens. Property \ref{regionV} ensures: all nets have the same initial sum of tokens.
\end{definition}

Obviously, the goal of Definition \ref{defRegion} is to define a distribution of tokens on a set of labelled nets such that equally labelled transitions have the same effect on the overall distribution (same label, same rise). If this is the case, for each such distribution, we can directly define a related place in the Petri net to be constructed, such that the arc weights connecting this place to all the transitions directly stem from the effect of every label on the overall distribution. The initial sum of tokens directly defines the initial marking of such a valid place. By construction, the region is a valid token trail for such a place. Thus, adding only places derived from regions to the synthesis result will guarantee \ref{synthA}.

\begin{theorem}
\label{thmOnePlace}
Let $S = \{(C_1,E_1,F_1,i_1,l_1),\ldots,(C_n,E_n,F_n,i_n,l_n)\}$ be a set of marked labelled nets, let $r$ be a region in $S$. We denote $E = \bigcup_i E_i$ and $l = \bigcup_i l_i$. For every $t \in \{l(e) | e \in E\}$ we fix one $e_t \in E$ so that $l(e_t) = t$ and $in_r(e_t) = min\{in_r(e) | l(e) = t\}$.
The Petri net $N = (P,T,W,m_0)$, where $P = \{p\}$, $T = \{l(e) | e \in E\}$, $W = \sum_{t \in T}in_r(e_t) \cdot (p,t) + \sum_{t \in T} out_r(e_t)\cdot(t,p)$, and $m_0(p) = \sum_{c \in C_1} i_1(c) \cdot r(c)$, is well-defined. For every $(C_\nu,E_\nu,F_\nu,i_\nu,l_\nu) \in S$, $r|_{C_\nu}$ is a valid token trail for $p$ of $N$. Thus, $S \subseteq \mathcal{L}(N)$ holds.
\end{theorem}
\begin{proof}
Fix some $\nu$, we prove $r_\nu := r|_{C_\nu}$ is a valid token trail in $(C_\nu,E_\nu,F_\nu,i_\nu,l_\nu)$ for $p$ of $N$.
\noindent
$\forall e \in E_\nu : W(p,l(e)) = in_r(e_{l(e)}) \leq in_{r_\nu}(e)$ so that \ref{ttI} holds.
\noindent
$\forall e \in E_\nu : W(l(e),p) - W(p,l(e)) = out_r(e_{l(e)}) - in_r(e_{l(e)}) = rise_r(e_{l(e)}) = rise_{r_\nu}(r)$ so that \ref{ttII} holds.
\noindent
$m_0(p) = \sum_{c \in C_1} i_1(c) \cdot r(c) = \sum_{c \in C_\nu} i_\nu(c) \cdot r_\nu(c)$ so that \ref{ttIII} holds as well.
The region $r$ is a valid token trail in every labelled net of $S$ for $p$ of $N$. Thus, $S \subseteq \mathcal{L}(N)$ holds.
\end{proof}

\begin{figure}[h]
    \centering
    \includegraphics[width=0.9\linewidth]{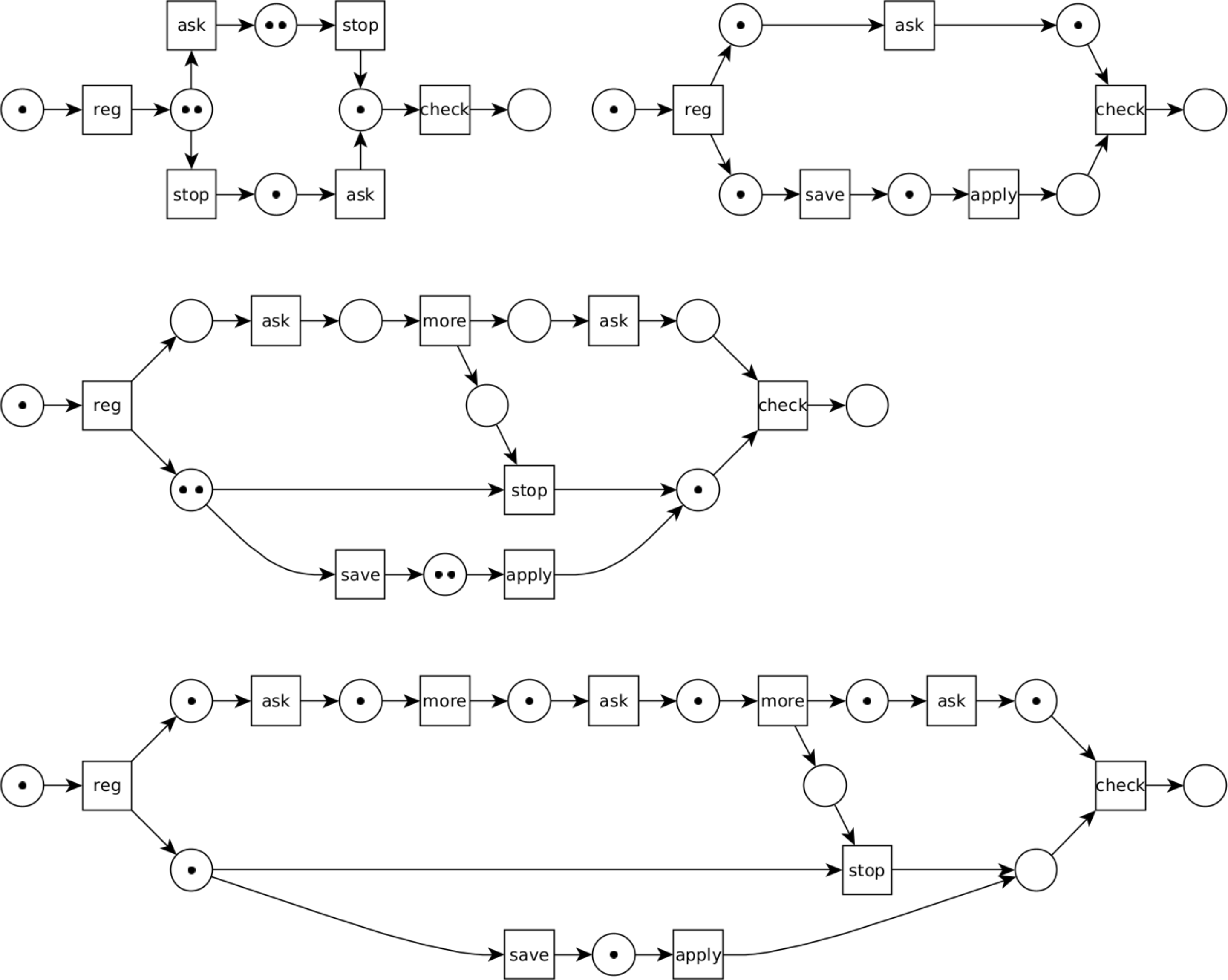}
    \caption{A token trail region in Figure \ref{netSpec}.}
    \label{region}
\end{figure}

\begin{figure}[h]
    \centering
        \includegraphics[
        width=0.75\linewidth,
        trim=0cm .5cm 0cm 1cm,
        clip
    ]{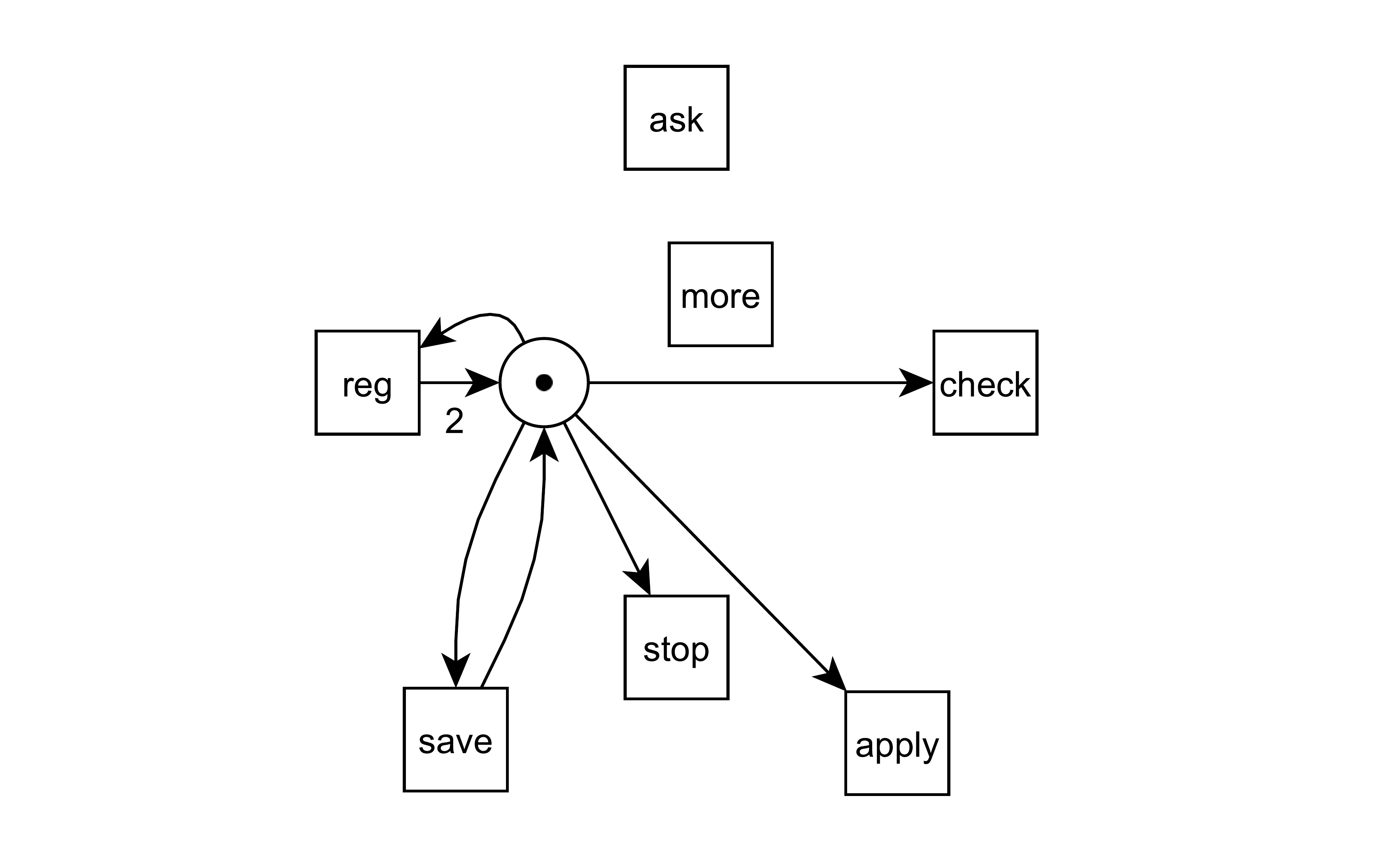}
    \caption{One-place Petri net defined by the region of Figure \ref{region}.}
    \label{onePlace}
\end{figure}

Theorem \ref{thmOnePlace} shows that whenever we have a region, we can construct a related one-place Petri net so that the region is a valid token trail for the place and thus, the net language includes the specification. Figure \ref{region} depicts markings in our specification so that \ref{regionIV} and \ref{regionV} hold. All initially marked places of Figure \ref{netSpec} have one token each. All transitions labelled \textit{reg} have rise $1$. All transitions labelled \textit{stop}, \textit{apply} or \textit{check} have rise $-1$. All other transitions have rise $0$.

Following Theorem \ref{thmOnePlace}, this region is a valid token trail for the place depicted in Figure \ref{onePlace}. The place is initially marked by one token and firing \textit{reg} will consume this token before producing two tokens in this place. Firing \textit{stop}, \textit{apply}, or \textit{check} will consume one token. Transition \textit{save} has a short loop to this place, and all other transitions are simply not connected. Adding this place to Figure \ref{petriNet}, the net language of the resulting net will still include Figure \ref{netSpec}.

\begin{corollary}
\label{corNetUnion}
Let $S$ be a set of labelled nets and $R$ be a set of regions in $S$. As defined in Theorem \ref{thmOnePlace}, every region $r \in R$ defines a one-place Petri net $N_r$. Let $N$ be the union of all Petri nets $N_r$, $S \subseteq \mathcal{L}(N)$ holds because of Definition \ref{defTokenTrails}.
\end{corollary}

\newpage
Corollary \ref{corNetUnion} provides the first puzzle piece to solve the synthesis problem. We add only places related to regions to our synthesis result to ensure \ref{synthA}. In the next step, we show that we never have to add a place that is not defined by a region. But first, we have to talk about short loops. Figure \ref{shortLoop} highlights the idea of the next theorem. In the first row, there is a Petri net and a labelled net of the net language of the Petri net. Thus, there is a valid token trail for $p$, which is depicted in the second row of the figure. This token trail naturally forms a region defining the Petri net of the second row with matching arc weights. Every place defined by a region maximizes arc weights but still the token trail for $p$ is also a valid token trail for $p'$. $p'$ is just more restrictive than $p$, and there is no need to add places not directly defined by regions.

\begin{figure}[h]
    \centering
    \includegraphics[width=0.7\linewidth]{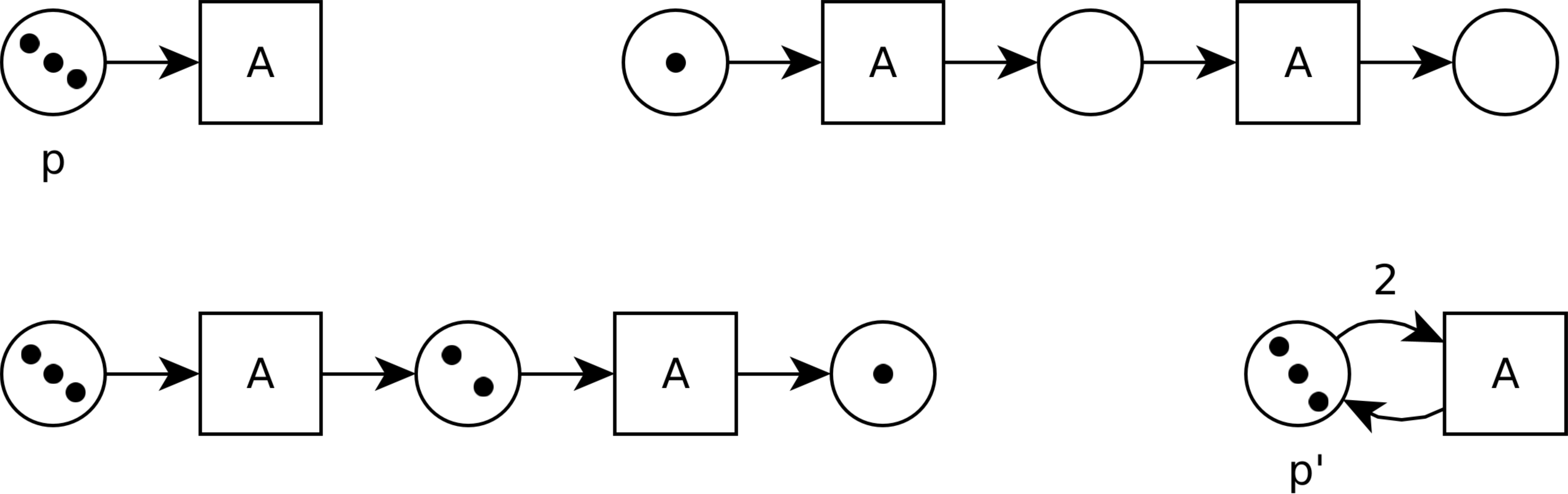}
    \caption{A Petri net, a marked labelled net, a valid token trail, and a one-place Petri net.}
    \label{shortLoop}
\end{figure}

\begin{theorem}
\label{thmMostRestrictive}
Let $S = \{(C_1,E_1,F_1,i_1,l_1),\ldots,(C_n,E_n,F_n,i_n,l_n)\}$ be a set of marked labelled nets. Let $N = (\{p\},T,W,m_0)$ be a one-place Petri net. If $S \subseteq \mathcal{L}(N)$ holds, there is a region $r$ in $S$ defining a one-place net $N' = (\{p'\},T,W',m'_0)$ so that $S \subseteq \mathcal{L}(N') \subseteq \mathcal{L}(N)$ holds.
\end{theorem}

\begin{proof}
If $S \subseteq \mathcal{L}(N)$ holds, there is a valid token trail for $p$ in every labelled net in $S$. Together, these markings form a region $r$ in $S$ because: \ref{ttIII} holds for every labelled net so that \ref{regionV} holds as well. \ref{ttII} holds for all the labelled transitions so that \ref{regionIV} holds as well. Because of \ref{ttI}, in every token trail, every labelled transition $t$ has at least $W(l(t),p)$ ingoing tokens, but there might be more. According to Theorem \ref{thmOnePlace}, we construct the related one-place net $(\{p'\},T,W',m'_0)$. To define $W'$. for every $t$, we take the minimal sum of ingoing tokens of all transitions labelled $t$. Thus, for every $t$, $W(p',l(t)) \geq W(p,l(t))$ holds. $r$ is a valid token trail for $p'$ (Theorem \ref{thmOnePlace}). Thus, every valid token trail for $p'$ is also a valid token trail for $p$ so that $S \subseteq \mathcal{L}(N') \subseteq \mathcal{L}(N)$ holds. Place $p'$ is just $p$ where short loops are maximized to get the most restrictive place.
\end{proof}

By combining Theorem \ref{thmOnePlace} and Theorem \ref{thmMostRestrictive}, we conclude that to solve the synthesis problem, we calculate regions and combine the corresponding one-place nets. We never have to add a place not defined by a region. 

\begin{corollary}
Let $S$ be a set of labelled nets and $R$ be a set of regions in $S$. Each region $r$ in $R$ defines a one-place Petri net $N_r$. The possibly infinite union net $N$ of the one-place Petri nets $N_r$ solves the synthesis problem.
\end{corollary}
\begin{proof}
Assume we have the possibly infinite union net $N$ of the one-place Petri nets $N_r$ and a place $p$, not yet in $N$. If adding $p$ to $N$ would still ensure \ref{synthA}, there is a valid token trail for $p$ in $S$. This token trail defines a region and thus, a related place $p'$ is already in $N$. There is no need to add $p$ because of Theorem \ref{thmMostRestrictive}.
\end{proof}

In praxis, to solve the synthesis problem, we need to construct a finite Petri net that has the same behaviour as the union Petri net. Upon reviewing the literature related to region theory, there are various approaches to solving this final step. For most state-based regions, when constructing Petri nets without arc weights, the number of regions is simply finite. If we consider state-based regions and Petri nets with arc weights, we typically attempt to solve the synthesis problem under the additional constraint that no extra behaviour is introduced. If such a result exists, the state space of the resulting net is the specification, and thus, only minimal regions need to be added. In the case of language-based regions, we aim to achieve the best upper approximation of the specified behaviour. Therefore, even if a region is larger in terms of specified behaviour, it may be smaller regarding unspecified behaviour. Thus, minimal regions might not be sufficient. Here, we calculate the set of wrong continuations, which is the boundary between unspecified and specified behaviour. For partial order-based specifications, this set of wrong continuations is finite, allowing us to determine whether there is an exact solution to the synthesis problem or not. If possible, excluding all wrong continuations while still ensuring \ref{synthAabstract} solves the synthesis problem.

In this paper, we define regions for a set of labelled nets. Thus, the behaviour is specified in terms of the net language of a Petri net. A net language is always infinite. Adding superfluous additional places, that don’t restrict the behaviour or simply copying a net and putting the copy in conflict to the original net, is, and should be, always possible. Thus, because a specification is always finite, and every result will have an infinite net language, technically there is never an exact solution.
Altogether, the number of regions is infinite because of arc weights, and the number of wrong continuations is infinite because it’s a net language. Thus, still the saturated net having places related to all regions is a solution, but a finite subnet might not be. 
Although, this sounds like bad news, we are not really interested in the complete, infinite net language of some Petri net. We want to model behaviour in terms of state graphs, partial languages and labelled nets and have a proper synthesis result that can simulate the specification but still falls into some reasonable net class. In this paper, we show that a very practical simplification will lead to very good results. As done for state-based regions, we simply restrict the size of every local marking of a region by some number $k$. Thus, the number of regions is again finite. In this procedure, every solution ensures \ref{synthA}, and there is no $k$-bounded solution having less behaviour. Additional behaviour decreases as we increase $k$.

\section{Token Trail Regions as Meta Regions}

In this section, we show that token trail regions form a meta theory for state-based and language-based regions.

It is easy to translate a state graph into a labelled net. Each state becomes a place, and each state-transition becomes a labelled transition. The place corresponding to the initial state carries one token in the initial marking. By construction, the reachability graph of the resulting net is the original state graph. The first net in Figure \ref{netSpec} shows such a labelled net. Its reachability graph is the net itself.

\begin{theorem}
\label{thmStateRegions}
Let $S = (R,T,X,i)$ be a state graph and $r$ be a region in $S$. We construct a labelled net $N = (C,E,F,m_0,l)$ by $C = R$, $E = X$, $F = \sum_{(s,t,s') \in X} (s,(s,t,s')) + ((s,t,s'),s')$, $m_0 = i$, and $l$ maps every $(s,t,s')$ to $t$. $r$ is a token trail region in $(C,E,F,m_0,l)$.
\end{theorem}
\begin{proof}
$r$ is a region in $S$ so that $r$ is a subset of $R$ so that every equally labelled state-transition either enters, exists, or does not cross $r$. If a state-transition $(s,t,s')$ exits $r$, $s \in r$ and $s' \not \in r$ hold. In $N$, $r$ is a marking, $s$ and $s'$ are places. $r$ marks $s$ with a single token and $s'$ with none. In $N$, $(s,t,s')$ is a transition connected only to $s$ and $s'$. The rise of $(s,t,s')$ in $r$ is $-1$. Using the same arguments, if a state-transition $(s,t,s')$ enters $r$, the rise of $(s,t,s')$ in $r$ is $1$. If a state-transition $(s,t,s')$ is non-crossing, the rise of transition $(s,t,s')$ in $r$ is $0$. Equally labelled state-transitions have the same crossing behaviour in $S$, so that equally labelled transitions have the same rise in $N$, i.e. \ref{regionIV} holds. \ref{regionV} holds, because we only consider one net with an initial sum of tokens $1$.
\end{proof}

Obviously, we can extend Theorem \ref{thmStateRegions} to state-based regions with arc weights. Here, regions are multisets of states, which again perfectly map to markings of a token trail region. The left side of Figure \ref{stateRegion} shows the input to a state-based synthesis algorithm and a region using a multiset of states. Lightly shaded states are in the region once, while darkly shaded states are in the region twice. \textit{reg} enters the region once, \textit{stop} exits the region once, and all other state-transitions do not cross. On the right, there is the related labelled net with the corresponding token trail region. It is simply the same concept.
Assume we have a labelled net, which directly models a state graph, and a token trail region. Using the same relation as defined in Theorem \ref{thmStateRegions}, highlighted in Figure \ref{stateRegion}, the token trail region in the labelled net corresponds to a region of the state graph. Considering state graphs, every region is a token trail region and vice versa.

\begin{figure}[h]
    \centering
    \includegraphics[width=\linewidth]{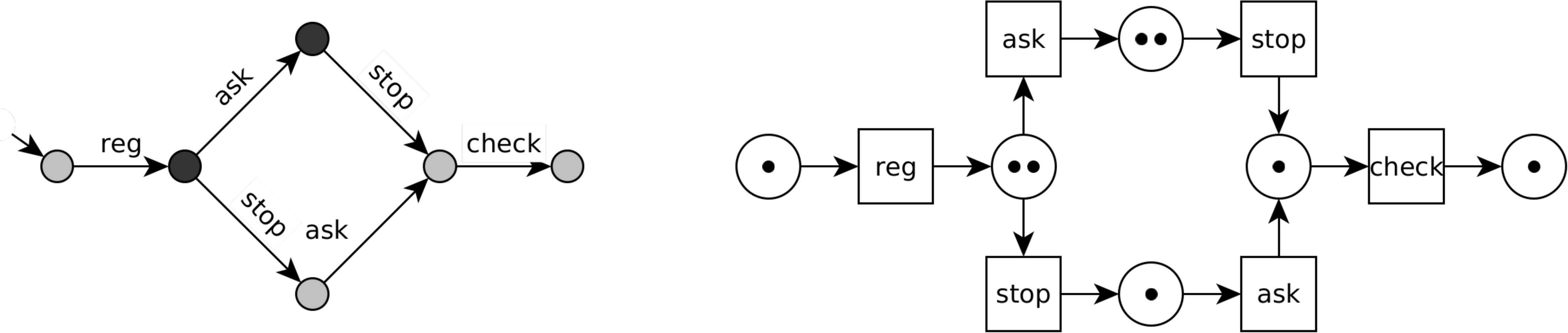}
    \caption{A state-based region and an identical token trail region.}
    \label{stateRegion}
    \vspace{-5mm}
\end{figure}

In Theorem \ref{thmLanguageRegions}, we obtain the same result for language-based regions. Here, we make use of the one-to-one correspondence between token trails and token flows on partial languages. This result was already introduced in \cite{DBLP:conf/apn/BergenthumFK23}. We repeat the theorem to also demonstrate that both related region definitions are equivalent.

\begin{theorem}
\label{thmLanguageRegions}
Let $S = \{(V_1,\ll_1,l_1),\ldots,(V_n,\ll_n,l_n)\}$ be a set of labelled partial orders and $r$ be a compact token flow region in $S$. We construct a set of labelled nets $N = \{(C_1,E_1,F_1,m_1,l_1), \ldots$, $ (C_n,E_n,F_n,m_n,l_n)\}$ by $C_i = (\{\blacktriangleright\}\times V_i)\cup{\ll_i} \cup (V_i \times\{{\scriptstyle\blacksquare}\})$, $E_i = V_i$, $F_i = \sum_{(v,v')\in {\ll_i}}(v,(v,v')) + ((v,v'),v') + \sum_{v \in V_i} ((\blacktriangleright,v),v) + (v,(v,{\scriptstyle\blacksquare}))$, and $m_i = \sum_{v \in V_i}(\blacktriangleright,v)$. $r$ is a token trail region in $N$.
\end{theorem}
\begin{proof}
We prove the one-to-one relation between token flows and token trails on partial languages. The only difference between \ref{ctfi}, \ref{ctfii}, \ref{ctfiii} and \ref{ttI}, \ref{ttII}, \ref{ttIII} is the definition of $in$  and $out$. In Definition \ref{defTokenTrails}, the sums are weighted by the corresponding arc weights. However, we do not need arc weights to model a labelled partial order as a labelled net. In this case, every weight is $1$. In Definition \ref{defTokenFlows} initial and superfluous tokens can be consumed and produced at every event. In this theorem, we add matching places to the preset and the postset of every labelled partial order. Thus, $r$ is a token trail region for $S$. 
\end{proof}

\begin{figure}[h]
    \centering
    \includegraphics[width=\linewidth]{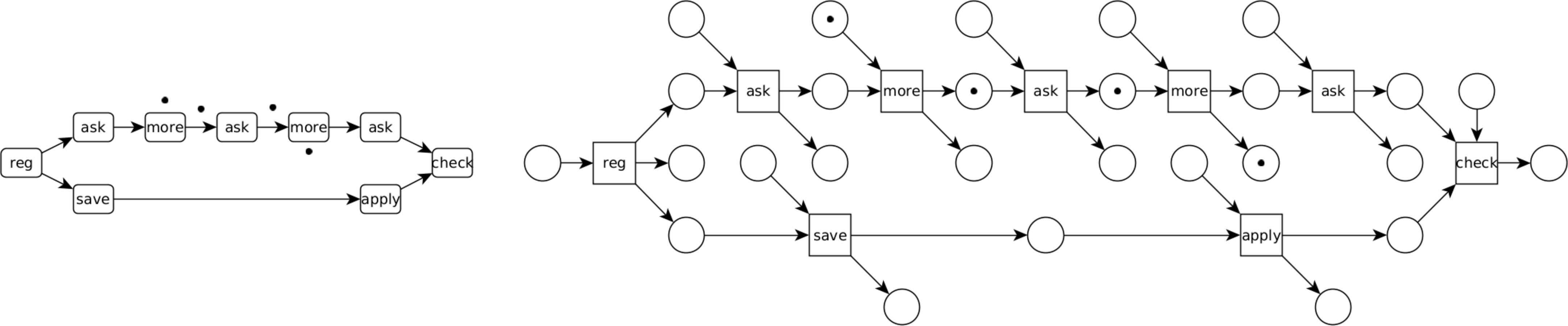}
    \caption{A language-based region and an identical token trail region.}
    \label{languageRegion}
\end{figure}

Figure \ref{languageRegion} depicts the relation between a valid compact token flow and a valid token trail on a partially ordered run. It is easy to see, that we could remove a lot of the additional places modelling the initial and final flow. The net is a partial order, so that tokens in the initial marking could spawn in the minimal places and be passed to every event. Similarly, superfluous tokens can be passed to the maximal places. Thus, if there is a token trail on such an extended net, there is also a token trail on the net only modelling the partial order adding minimal and maximal places for the initial and final marking. See the second net of Figure \ref{region} as an example.

Assume we have a set of labelled nets, which directly model a partial language, and a token trail region. Using the same relation as defined in Theorem \ref{thmLanguageRegions}, highlighted in Figure \ref{languageRegion}, the token trail region in the labelled net defines a region of the partial language. Considering partial languages, every region is a token trail region and vice versa.

Concluding this section, as conjectured in \cite{DBLP:conf/ataed/BergenthumK22}, token trail regions are a meta region definition. Instead of state-based regions, we can simply use token trail regions. The same applies to partially ordered regions. Furthermore, we can use minimal regions or wrong continuations to obtain a finite result. However, the goal of this paper is not only to unify both concepts under a meta region definition, but also to introduce a new method for practical applications. By using the new definition, we are not only able to combine state-based and language-based input, but we can also incorporate labelled nets into our specification. These labelled nets can feature distributed initial markings, loops, and even arc weights. In the following, we present an approach for calculating a set of minimal regions up to a bound of $k$, in order to address the synthesis problem from a set of labelled nets.

\section{Computing Token Trail Regions}
\label{sec4}
In this section, we introduce a method to calculate all minimal token trail regions up to a given bound $k$. We present an implementation, a web-tool, and provide examples of the synthesis results to demonstrate that the synthesis is indeed applicable in practical scenarios.

Typically, state-based regions are constructed using a bottom-up approach. Starting from a set of so-called excitation regions, these regions are gradually repaired by incrementally adding more states until they form a set of valid regions. The bottom-up approach generates a set of minimal regions. A similar process could be applied to token trail regions. However, the bottom-up approach can become somewhat unwieldy due to the concurrency in the labelled nets, which introduces many  alternatives when repairing regions.

Most synthesis algorithms based on languages use Integer Linear Programming (ILP) combined with the wrong continuation approach to identify the set of possible regions. For example, see \cite{DBLP:conf/apn/vanDerWerfDHS08} for the most prominent algorithm that applies language-based regions, ILP, and a heuristic based on the directly-follows relation of the input to generate a finite set of regions. The Kokosminer, based on the Master thesis of Karl Heinrichmeyer \cite{heinrichmeyer20}, directly calculates a set of minimal regions with an ILP. In our implementation, we generalize Karl’s approach to calculate token trail regions.

An ILP is easy to implement using specialized solvers and is highly flexible. It allows for intuitive implementation of restrictions on the type of regions we aim for, such as one-bounded regions, or the type of net class we are targeting, for example, nets with or without arc weights and/or short loops. For these reasons, we combine the two best proven techniques and introduce a new ILP approach for calculating minimal regions in this paper.

\newpage
To calculate regions, formally a multiset of places, we implement the conditions of Definition \ref{defRegion} as an ILP. Every place of the specification is a variable of the ILP. For every label of the specification, we fix one transition with the related label. For every other transition carrying the same label, we introduce one equality to the ILP ensuring the rise of both transitions is equal. Let’s consider Figure \ref{netSpec} as an example. The specification has $35$ places, thus, the ILP has $35$ variables. There are four transitions labelled \textit{reg}. Let us call them \textit{reg\textsubscript{1}}, \textit{reg\textsubscript{2}}, \textit{reg\textsubscript{3}}, and \textit{reg\textsubscript{4}}. We add three equalities to the ILP. One ensures that the rise of \textit{reg\textsubscript{1}} is equal to the rise of \textit{reg\textsubscript{2}}. Another, ensuring that the rise of \textit{reg\textsubscript{1}} is equal to the rise of \textit{reg\textsubscript{3}}. And one more, ensuring the rise of \textit{reg\textsubscript{1}} is equal to the rise of \textit{reg\textsubscript{4}}. There are eight transitions labelled \textit{ask}, thus, we add seven more equalities. Altogether we add these equalities for every label. They ensure that every solution of the ILP satisfies \ref{regionIV}.

In the next step, we add equalities to ensure \ref{regionV}. Again, we simply refer to all labelled nets of the specification as \textit{net\textsubscript{1}}, \textit{net\textsubscript{2}}, \dots , \textit{net\textsubscript{n}}. We construct one equality such that the sum of the initial tokens of \textit{net\textsubscript{1}} and \textit{net\textsubscript{2}} is equal. Another equality ensures that the sum of the initial tokens of \textit{net\textsubscript{1}} and \textit{net\textsubscript{3}} is equal, and so on. They ensure that every solution of the ILP satisfies \ref{regionV}.

We add a constraint to each variable, ensuring that it is less than or equal to $k$, in order to calculate regions only up to $k$.

Altogether, every non-negative integer solution of this region ILP is a token trail region up to $k$, and vice versa. This construction is really kind of straight forward, because Definition \ref{defRegion} is an ILP already.

In the next step, we search for a non-empty minimal region. We add the inequality that the sum of all variables is greater than $0$, and we also add the same sum, i.e., the sum of all variables, as the objective function to be minimized by the ILP. Now, every optimal solution corresponds to a minimal region, because there is no other region that has fewer or an equal number of tokens in every place, the objective function of such region would be smaller.

With the first minimal region, we build a transition for every label, along with arcs and a place, as defined by Theorem \ref{thmOnePlace}. Now, we have to find other minimal regions and related places.

Every region we find is a solution $(s_1, \ldots, s_n)$ assigning values to the variables $p_1, \ldots, p_n$ of the ILP. With every new solution, we extend the ILP. Let’s assume we just found a minimal region $r_1 = s_1 \cdot p_1 + \ldots + s_n \cdot p_n$. For every multiplicity $s_i > 0$ of $r_1$, we add a variable $x_{r1pi}$ to the ILP. The idea is to add inequalities to the ILP so that $x_{r1pi}$ will be $1$ if the variable $p_i$ is less than $s_i$, else $0$. Therefore we add the following inequalities:

\begin{centering}
\vspace{2mm}
$0 \leq (p_1 - s_1) + k \cdot x_{r1p1} \leq k - 1$

$\ldots$

$0 \leq (p_n - s_n) + k \cdot x_{r1pn} \leq k - 1$

\vspace{2mm}
\end{centering}

If variable $p_i$ is smaller than the previous solution $s_i$, the term $(p_i - s_i)$ is less than 0, so that $x_{r1pi}$ must be $1$. If $p_i$ is at least $s_i$, the term $(p_i - s_i)$ is non-negative and $x_{r1pi}$ must be $0$. Finally, we add $x_{r1p1} + \ldots + x_{r1pn} > 0$ to our ILP to ensure, that every new solution will be smaller in at least one component.

By iteratively building, solving, and updating the ILP, we can calculate the set of minimal regions up to $k$ for our specification, and subsequently construct the related Petri net.

This synthesis approach is implemented as a new module of the I \emojiHeart Petri nets website. The \emojiHorse module, available at \url{www.fernuni-hagen.de/ilovepetrinets/horse}, implements the algorithm introduced in this chapter by solving a sequence of ILPs using the GLPK-Solver \cite{glpk,glpk.js}. For this contribution, we prepared four examples that the reader can try to solve using our new region approach.

We already mentioned Examples 1 to 3 in the introduction. Example 1, the net based-specification of Figure \ref{netSpec}, synthesises Figure \ref{petriNet} in $30$ seconds. Example 2, the state-based specification of Figure \ref{stateGraph}, synthesises Figure \ref{petriNet}, without $p_5$, instantly. State-based input cannot specify the short loop. Example 3, four runs of the language-based specification of Figure \ref{pos}, synthesises Figure \ref{petriNet} in $6$ minutes. Language-based regions can't handle all six in reasonable time. Example $4$ highlights another application of the token trail approach. Obviously, every Petri net is also a labelled net. Thus, the input to the synthesis can be a fully-fledged process model. In Example $4$, we input two general labelled nets, one with a distributed initial marking, the other with an arc weight. Here we synthesise Figure \ref{petriNet} in no-time. This is an exciting feature of token trail regions. We can input sequences, state graphs, partially ordered runs, workflow nets, and complete system models to synthesise a Petri net that simulates the specified behaviour. 

We encourage the reader to try more examples. If we delete the longest run in Example 1, where \textit{more} occurs twice, the result will have no arc weights. If we input an unlabeled net, the result will be identical to the input. We can add the two nets from Example 4 to Example 1 without changing the result. Inputs can be uploaded using the \textit{PNML} standard or created with an editor available at the I \emojiHeart Petri Nets website, available at \url{www.fernuni-hagen.de/ilovepetrinets/}.

\section{Process Discovery using Token Trail Regions}

In the previous sections, we introduced token trail regions as a unifying region theory for synthesizing a Petri net from a behavioural specification given in terms of a set of labelled nets. While the focus so far has been on the theoretical foundations and the synthesis problem itself, we now turn to applications. In particular, we discuss process discovery, which is the most prominent practical application of region-based synthesis techniques.

Process discovery \cite{DBLP:journals/topnoc/AalstD13,DBLP:books/sp/Aalst16,DBLP:books/sp/22/PMH2022} aims at automatically generating a process model from observed example behaviour. In most practical settings, such example behaviour is recorded in the form of an event log. An event log consists of a multiset of traces, where each trace is a finite sequence of activity labels describing one observed execution of the process. Depending on the level of abstraction, additional information such as timestamps, resources, or data attributes may be available, but classical discovery focuses on the control-flow perspective.

Among the region-based discovery approaches, the ILP Miner \cite{DBLP:conf/apn/vanDerWerfDHS08} is the most prominent representative. It applies language-based regions to the prefix-closed language induced by the event log and formulates the synthesis problem as an integer linear program in order to compute places of a workflow net. Region-based discovery algorithms such as \cite{DBLP:conf/bpm/CarmonaCK08,DBLP:conf/icpm/Bergenthum19, DBLP:conf/apn/FolzWeinsteinBDK23,DBLP:conf/caise/FolzWeinsteinRMBA25} are known to produce precise and well-fitting models, in contrast to heuristic approaches that often trade precision for performance. However, these advantages come at a cost. The runtime of synthesis-based discovery algorithms can be high, and the quality of the result heavily depends on the quality of the input log.

From a synthesis perspective, region-based techniques always assume that the observed behaviour is sufficiently complete and correct. That is, the input behaviour is implicitly treated as a specification rather than as a noisy or incomplete sample. If relevant behaviour is missing from the log, or if undesired behaviour is present due to noise, region-based synthesis tends to either overgeneralise or fail to produce a meaningful model. This observation also applies to our contribution. Token trail regions are particularly well suited when the observed behaviour can be considered a specification of the intended system behaviour.

Token trail regions provide a natural generalisation of region-based process discovery. Instead of restricting the input to sequences or partial orders derived from an event log, we allow the specification to be given as a set of labelled nets. This enables several extensions of classical discovery scenarios.

First, classical event logs can be handled directly. Every trace of an event log can be translated into a simple sequential labelled net. Alternatively, traces can be grouped and folded into labelled nets that share common prefixes, suffixes, or local states. This can significantly reduce the size of the input specification, thereby improving runtime.

Second, token trail regions allow us to incorporate additional behavioural knowledge beyond raw event data. For example, concurrency relations derived from timestamps, domain knowledge, or prior analysis can be encoded explicitly using labelled nets. Similarly, known conflicts, loops, or synchronisation patterns can be specified using small labelled nets and added to the specification. This is particularly useful in process discovery settings where logs are incomplete, but partial models or constraints are available.

Third, token trail regions support the combination of heterogeneous sources of behavioural information. A specification may contain traces extracted from an event log, partial orders obtained from concurrency analysis, state-based fragments derived from legacy systems, and even existing process models. All these artefacts can be represented uniformly as labelled nets and jointly processed by the synthesis algorithm. The resulting Petri net simulates all provided inputs and therefore integrates the different perspectives into a single model.

From a computational point of view, the ILP-based computation of token trail regions closely resembles existing discovery algorithms. Bounding the local markings of regions by a parameter $k$ yields a finite and controllable search space. As in classical region-based discovery, increasing $k$ reduces additional behaviour, while smaller values of $k$ lead to more general models. This provides a natural mechanism to balance precision and generalisation.

In summary, token trail regions extend region-based process discovery in a conservative but powerful way. They subsume classical language-based discovery techniques, while enabling richer specifications that go beyond event logs. This makes token trail regions particularly suitable for discovery scenarios in which observed behaviour is complemented by structural knowledge, partial models, or manually specified behavioural fragments.

An additional advantage of token trail regions in the context of process discovery is their ability to eliminate label splitting in discovered process models. Several discovery algorithms, such as Alpha++ \cite{DBLP:journals/dmkd/WenWWWS07} and Split Miner \cite{DBLP:journals/tkde/AugustoCDR20}, explicitly allow label splitting in order to represent different behavioural contexts of activities. While label splitting can significantly improve the precision and structural clarity of discovered models, it comes at the cost of increased complexity. In particular, analysis techniques such as conformance checking, performance analyses, or further synthesis steps often suffer from higher runtime when applied to labelled nets. The reason is that there is no longer a one-to-one correspondence between paths in the process model and firing sequences. As a consequence, many analysis techniques have to operate on the reachability graph of the process model in such applications.

Using token trail regions, label splitting can be treated as a modelling convenience rather than a permanent feature of the resulting model. Behaviour that is discovered using labelled nets can be seen as a specification itself. Token trail synthesis ensures, that the behavioural distinctions introduced by label splitting are preserved semantically through places and arc weights, rather than through duplicated labels. In this way, token trail regions allow us to use labelled nets as an expressive intermediate representation during discovery, while still producing compact and efficiently analysable Petri net models as synthesis results.

Another interesting feature of the approach presented in this contribution is that token trail regions are computed using an ILP that explicitly enumerates a set of minimal regions. This allows us to control each place that is added to the synthesis result. For example, in the context of process discovery, we are often interested in synthesising workflow nets. A workflow net is a Petri net with a distinguished initial place~$i$ and a final place~$o$ such that every node of the net lies on a directed path from~$i$ to~$o$. For workflow nets, an additional desirable property is soundness~\cite{DBLP:journals/fac/AalstVHVSV11}. A workflow net is sound if, starting from the initial marking, it is always possible to eventually reach the final marking. Roughly speaking, a sound workflow net starts with a single token in~$i$, and an execution is complete if it reaches a single token in~$o$. Consequently, all inner places of a workflow net should be empty both at the start and at the end of every execution.

To tailor token trail synthesis towards the discovery of sound workflow nets, we first restrict our specifications to labelled workflow nets. For example, when translating an event log into a set of sequentially ordered labelled nets, each such trace net already has a unique initial place and a unique final place. In the same way, we can require that every partial order or labelled net included in the specification provides exactly one initial place and one final place. When computing regions for the inner places of the intended workflow net result, we require that the token trail of a region does not mark the final place of any labelled net in the specification. This yields a simple additional inequality for each input net, ensuring that for every added region the corresponding place is empty after every specified run. Consequently, executing any specified run will never leave a token in the synthesis result. Such places do not obstruct proper termination and therefore do not prevent the resulting net from being sound.

It is important to note that this restriction does not guarantee soundness of the synthesised net, as it does not ensure that the result is connected or free of deadlocks. However, it provides an effective mechanism to steer the synthesis towards the intended net class. By excluding places that may carry tokens at the end of the observed behaviour, the resulting model generalises the specification while remaining closer to sound workflow nets. 

As an illustration, the \emojiHorse\ tool provides an additional option that activates this ILP extension. By switching the mode from synthesis to discovery, the corresponding empty at the end constraints are enabled. Applying this setting to \textit{Example~1} yields the Petri net of Figure~\ref{petriNet}, but without the places $p_3$ and $p_5$. All four labelled nets of the specification are workflow nets. Therefore, we can easily add the constraint that the final place of each input net must be unmarked for every computed region. Place $p_3$ counts loop iterations and is therefore not empty after the first three runs of Figure~\ref{netSpec}. Place $p_5$ is empty only if transition \textit{stop} has occurred. This example demonstrates how token trail regions can be used to guide process discovery towards workflow nets with stronger termination properties.

With this technique, we obtain a Petri net that can simulate the specified behaviour, but that does not restrict it maximally. In the given example, the resulting net no longer represents the observed behaviour with the same precision as the unrestricted synthesis result. However, in a setting where the input consists only of a snippet of the complete behaviour, this yields a valid and sound solution that completes the observed behaviour. This illustrates a fundamental trade-off between behavioural precision and structural guarantees. Depending on the application scenario, token trail regions allow us to systematically steer the synthesis process. By strengthening or relaxing constraints in the ILP, we can move along a spectrum between behaviour-preserving synthesis and more generalising discovery, thereby tailoring the result towards the desired net class and analysis objectives.

\section{Conclusion}

In this paper, we introduced token trail regions as a new region definition for Petri net synthesis from behavioural specifications given in terms of labelled Petri nets. Token trail regions are based on the token trail semantics and define regions as label-respecting distributions of local states across a set of labelled nets. This allows us to synthesize Petri nets that simulate the specified behaviour while introducing as little additional behaviour as possible.

A central contribution of this work is the unification of state-based and language-based region theory. We formally showed that token trail regions subsume both classical state-based regions and language-based regions, and therefore form a meta region theory. As a consequence, behavioural specifications given as state graphs, languages, partial orders, or combinations thereof can be handled uniformly. Moreover, token trail regions extend this unification by allowing labelled Petri nets as first-class specification artefacts. This enables the explicit modelling of conflict, concurrency, and merging of local states within a single specification, without forcing the modeller to commit to one specific behavioural representation in advance.

From a modelling perspective, labelled nets provide a natural and intuitive way to specify behaviour. They allow modellers to decompose complex behaviour into smaller fragments, to use label splitting to unfold behaviour where necessary, and to combine different modelling styles depending on the nature of the available information. In contrast to more expressive behavioural formalisms such as branching processes or event structures, labelled nets do not introduce a new modelling language, but build directly on well-established Petri net concepts. Token trail regions preserve this modelling convenience while providing a rigorous synthesis semantics.

On the synthesis side, we showed that every token trail region defines a valid place in the synthesis result and that no additional places beyond those induced by regions are required. This result generalises classical region-based synthesis and ensures that synthesis can always be reduced to the computation of regions. Although the set of all regions is infinite in general, we demonstrated that restricting the size of local markings by a bound~$k$ yields a finite and practically manageable approximation. Increasing~$k$ monotonically reduces additional behaviour, providing a principled way to balance precision and generalisation.

We presented an ILP-based approach for computing minimal token trail regions up to a given bound and implemented this approach in a web-based tool. Our experimental results indicate that, despite the inherent complexity of region-based synthesis, the approach is practically feasible and performs well on a range of examples. In particular, token trail regions allow us to combine heterogeneous specifications, including state-based input, language-based input, workflow nets, and complete process models, into a single synthesis task.

We further demonstrated that token trail regions naturally support process discovery scenarios. By interpreting observed behaviour as a specification rather than a noisy sample, token trail synthesis provides precise discovery results when sufficient information is available. At the same time, the ILP formulation allows the synthesis process to be steered towards specific net classes, such as workflow nets with termination guarantees. This makes token trail regions well suited for discovery settings in which precision, soundness, and analysability must be carefully balanced.

Future work will focus on several directions. On the theoretical side, we plan to further investigate formal guarantees for the synthesis results, in particular with respect to soundness and behavioural minimality under additional constraints. On the algorithmic side, we aim to explore optimisation techniques for the ILP formulation and incremental approaches for handling larger specifications. From an application perspective, we intend to study the integration of token trail regions with noise-handling and abstraction techniques from process mining, in order to relax the assumption of complete and correct specifications. Altogether, token trail regions provide a flexible and unifying foundation for Petri net synthesis, bridging theory and practice and opening new avenues for modelling, synthesis, and discovery of concurrent systems.

%
%
\bibliographystyle{fundam}
\bibliography{paper}
\end{document}